\documentclass[aps,pra,11pt,letterpaper,onecolumn,tightenlines,superscriptaddress,notitlepage,longbibliography]{revtex4-2}
\usepackage{qcircuit}
\usepackage[dvips]{graphicx}
\usepackage{amsmath,amssymb,amsthm,mathrsfs,amsfonts,dsfont}
\usepackage{mathtools}
\usepackage[caption=false]{subfig}
\usepackage{braket}
\usepackage{bm}
\usepackage{enumerate}
\usepackage{xcolor}
\usepackage{comment}
\usepackage{scalefnt}
\usepackage[colorlinks = true, 
citecolor = blue,
linkcolor = blue,
urlcolor = blue]{hyperref}
\usepackage{makecell}

\usepackage{relsize}
\usepackage{pagecolor}
\usepackage[T1]{fontenc}
\usepackage{lmodern}
\usepackage[utf8]{inputenc}
\usepackage{microtype}

\usepackage{tikzscale}
\usepackage{pgfplots}
\pgfplotsset{compat=newest}
\usetikzlibrary{plotmarks}
\usetikzlibrary{arrows.meta}
\usetikzlibrary{external}
\usepgfplotslibrary{patchplots}
\usepackage{grffile}
\usepackage{enumitem}

\DeclareMathOperator{\Tr}{Tr}

\newcommand{\vspan}{\mathrm{span}}

\newcommand{\coleq}{\mathrel{\mathop:}\nobreak\mkern-1.2mu=}
\newcommand{\eqcol}{\mkern-1.2mu=\nobreak\mathrel{\mathop:}}

\newcommand{\od}[1]{^{(#1)}}

\newcommand{\pt}{\mathrm{pt}}
\newcommand{\mf}{\mathfrak}
\newcommand{\mc}{\mathcal}

\newcommand{\mr}{\mathrm}
\newcommand{\mbb}{\mathbb}

\newcommand{\expval}[1]{{\langle #1 \rangle}}

\newcommand{\E}{{\mathop{\mbb{E}}}}
\newcommand{\Pn}{{{\sf P}^n}}

\newtheorem{theorem}{Theorem}
\newtheorem{proposition}[theorem]{Proposition}

\newtheorem{lemma}{Lemma}

\newtheorem{corollary}[theorem]{Corollary}
\newtheorem{definition}{Definition}

\definecolor{applegreen}{rgb}{0.55, 0.71, 0.0}

\usepackage{algorithm}  
\usepackage[noend]{algpseudocode}  
\newcommand{\algorithmfootnote}[2][\footnotesize]{%
  \let\old@algocf@finish\@algocf@finish
  \def\@algocf@finish{\old@algocf@finish
    \leavevmode\rlap{\begin{minipage}{\linewidth}
    #1#2
    \end{minipage}}%
  }%
}
\usepackage{xparse}
\NewDocumentCommand{\LeftComment}{s m}{%
  \Statex \IfBooleanF{#1}{\hspace*{\ALG@thistlm}}\(\triangleright\) #2}
\algnewcommand{\LineComment}[1]{\Statex // #1}

\usepackage{footmisc}
\usepackage{cleveref}

\begin{document}
\title{The learnability of Pauli noise}
\author{Senrui Chen}
\thanks{S.C. and Y.L. contributed equally to this work (alphabetical order). Correspondence and requests for materials should be addressed to S.C. (\href{mailto:csenrui@uchicago.edu}{csenrui@uchicago.edu}), Y.L. (\href{mailto:yunchaoliu@berkeley.edu}{yunchaoliu@berkeley.edu}) or L.J. (\href{mailto:liang.jiang@uchicago.edu}{liang.jiang@uchicago.edu}).}
\affiliation{Pritzker School of Molecular Engineering, University of Chicago, IL 60637, USA}
\author{Yunchao Liu}
\thanks{S.C. and Y.L. contributed equally to this work (alphabetical order). Correspondence and requests for materials should be addressed to S.C. (\href{mailto:csenrui@uchicago.edu}{csenrui@uchicago.edu}), Y.L. (\href{mailto:yunchaoliu@berkeley.edu}{yunchaoliu@berkeley.edu}) or L.J. (\href{mailto:liang.jiang@uchicago.edu}{liang.jiang@uchicago.edu}).}
\affiliation{Department of Electrical Engineering and Computer Sciences, University of California, Berkeley, CA 94720, USA}
\author{Matthew Otten}
\affiliation{HRL Laboratories, LLC, 3011 Malibu Canyon Rd., Malibu, CA 90265, USA}
\author{Alireza Seif}
\affiliation{Pritzker School of Molecular Engineering, University of Chicago, IL 60637, USA}
\author{Bill Fefferman}
\affiliation{Department of Computer Science, University of Chicago, IL 60637, USA}
\author{Liang Jiang}
\thanks{S.C. and Y.L. contributed equally to this work (alphabetical order). Correspondence and requests for materials should be addressed to S.C. (\href{mailto:csenrui@uchicago.edu}{csenrui@uchicago.edu}), Y.L. (\href{mailto:yunchaoliu@berkeley.edu}{yunchaoliu@berkeley.edu}) or L.J. (\href{mailto:liang.jiang@uchicago.edu}{liang.jiang@uchicago.edu}).}
\affiliation{Pritzker School of Molecular Engineering, University of Chicago, IL 60637, USA}

\date{\today}
\begin{abstract}
Recently, several quantum benchmarking algorithms have been developed to characterize noisy quantum gates on today’s quantum devices. A fundamental issue in benchmarking is that not everything about quantum noise is learnable due to the existence of gauge freedom, leaving open the question what information is learnable and what is not, which is unclear even for a single CNOT gate. Here we give a precise characterization of the learnability of Pauli noise channels attached to Clifford gates using graph theoretical tools. Our results reveal the optimality of cycle benchmarking in the sense that it can extract all learnable information about Pauli noise. We experimentally demonstrate noise characterization of IBM’s CNOT gate up to 2 unlearnable degrees of freedom, for which we obtain bounds using physical constraints. In addition, we show that an attempt to extract unlearnable information by ignoring state preparation noise yields unphysical estimates, which is used to lower bound the state preparation noise.
\end{abstract}
\maketitle

\section{Introduction}
Characterizing quantum noise is an essential step in the development of quantum hardware~\cite{eisert2020quantum,preskill2018quantum}. Remarkably, despite recent progress in both gate-level and scalable noise characterization methods~\cite{Emerson2005scalable,Knill2008randomized,Dankert2009exact,Magesan2011scalable,Magesan2012characterizing,helsen2020general,erhard2019characterizing,flammia2020efficient,harper2020efficient,harper2020fast,flammia2021pauli,liu2021benchmarking,flammia2021averaged,chen2022quantum}, the full characterization of the noise channel of a single CNOT/CZ gate remains infeasible. This is unlikely to be caused by limitations of existing benchmarking algorithms. Instead, it is believed to be related to the fundamental question of what information about a quantum system can be learned, in a setting where initial states, gates, and measurements are all subject to unknown quantum noise. It is well-known that \emph{some} information about quantum noise can be learned (such as the gate fidelity learned by randomized benchmarking~\cite{Emerson2005scalable,Knill2008randomized,Dankert2009exact,Magesan2011scalable,Magesan2012characterizing} or cycle benchmarking~\cite{erhard2019characterizing}), but \emph{not everything} can be learned (due to the gauge freedom in gate set tomography~\cite{Merkel2013Self-consistent,Blume-Kohout2013Robust,nielsen2021gate}). The boundary of learnability of quantum noise -- a precise understanding of what information is learnable and what is not, still remains an open question.

Recently, there has been an interest in formulating noise characterization as learning unknown gate-dependent Pauli noise channels~\cite{erhard2019characterizing,harper2020efficient}. This is motivated by randomized compiling, a technique that has been proposed to suppress coherent errors via inserting random Pauli gates~\cite{wallman2016noise,hashim2020randomized}. As an added benefit, randomized compiling twirls the gate-dependent CPTP noise channel into Pauli noise, thus reducing the number of parameters to be learned. Note that the twirled Pauli noise channel corresponds to the diagonal of the process matrix of the CPTP map, so Pauli noise learning is a necessary step for characterizing the CPTP map, regardless of whether randomized compiling is performed.

However, even under this simplified setting of Pauli noise learning, all prior experimental attempts can only partially characterize the noise channel of a single CNOT/CZ gate~\cite{hashim2020randomized,berg2022probabilistic,Ferracin2022Efficiently}, which only has 15 degrees of freedom. A natural question is whether this limitation is caused by the fundamental unlearnability of the noise channel, and if so, which part of the noise channel and how many degrees of freedom among the 15 are unlearnable?
        
In this paper, we give a precise characterization of what information in the Pauli noise channel attached to Clifford gates is learnable, in a way that is robust against state preparation and measurement (SPAM) noise. 
We develop a systematic method for characterizing learnable degrees of freedom of a Clifford gate set using notions from algebraic graph theory
and show that learnable information exactly corresponds to the cycle space of the Pauli pattern transfer graph, while unlearnable information exactly 
corresponds to the cut space. This characterization can be used to write down a list of linear functions of the noise model that corresponds to all independent learnable degrees of freedom. As an example, we show that the Pauli noise channel of an arbitrary 2-qubit Clifford gate has at most 2 unlearnable degrees of freedom. We perform an experimental characterization of a CNOT gate on IBM Quantum hardware~\cite{ibmquantum} up to 2 unlearnable degrees of freedom. Although the unlearnable information cannot be estimated with high precision, we can determine a feasible region of those freedoms using the constraint that the noise model must be physical (\textit{i.e.}, all Pauli error rates are nonnegative).

A corollary of our result is that cycle benchmarking is optimal in the setting we consider, in the sense that it can learn all the information that is learnable. This reveals a fundamental fact about noise benchmarking, namely that cycle benchmarking -- the idea of repeatedly applying the same gate sequence interleaved by single qubit gates, is the ``right'' algorithm for benchmarking Clifford gates, because of the fact that learnable information forms a cycle space.
As an interesting side remark, the term ``cycle'' in cycle benchmarking originally refers to parallel gates applied in a clock cycle. Here we show that the term can also be understood in a graph-theoretical context.

In addition, we also explore ways to overcome the unlearnability barrier. It has been recognized that the unlearnability does not apply if the initial state $\ket{0}^{\otimes n}$ can be prepared perfectly~\cite{flammia2021averaged,Ferracin2022Efficiently}, and it has been suggested that state preparation noise could be much smaller than gate and/or measurement noise in practice~\cite{Maciejewski2020mitigationofreadout,Bravyi2021Mitigating,Ferracin2021Experimental}, which would make gate noise fully learnable up to small error. We develop an algorithm based on cycle benchmarking that fully learns gate-dependent Pauli noise channel assuming perfect initial state preparation, and experimentally demonstrate the method on IBM's CNOT gate. Based on the experiment data, we conclude that this assumption is unlikely to be correct in our experiment as it gives unphysical estimates that are outside the feasible region we determined. Furthermore, we use the data to obtain a lower bound on the state preparation noise and conclude that it has the same order of magnitude as gate noise on the device we used. Therefore, the issue of unlearnability is a practically relevant concern, for which the noise on initial states is an important factor that cannot be neglected on current quantum hardware.

\section{Results}
\subsection{Theory of learnability}
We start by considering the learnability of the Pauli noise channel of a single $n$-qubit Clifford gate. A Pauli channel can be written as 
\begin{equation}
    \Lambda(\cdot) = \sum_{a\in{{\sf P}^n}}p_a P_a(\cdot)P_a,
\end{equation}
where $\{p_a\}$ is a probability distribution on ${\sf{ P}}^n=\{I,X,Y,Z\}^n$. The goal is to learn this distribution, which has $4^n-1$ degrees of freedom. Considering $\Lambda$ as a linear map, its eigenvectors exactly correspond to all $n$-qubit Pauli operators, as
\begin{equation}
    \Lambda(P_a) = \lambda_a P_a,\quad\forall a\in{{\sf P}^n}
\end{equation}
where $\lambda_a = \sum_{b\in{{\sf P}^n}}p_b(-1)^\expval{a,b}$ is the Pauli fidelity associated with the Pauli operator $P_a$. Therefore $\Lambda$ is a linear map with known eigenvectors and unknown eigenvalues, so a natural way to learn $\Lambda$ is to first learn all the Pauli fidelities $\lambda_a$, and then reconstruct the Pauli errors via $p_a = \frac{1}{4^n}\sum_{b\in{{\sf P}^n}}\lambda_b(-1)^\expval{a,b}$.

\begin{figure}[t]
    \centering
    \subfloat[standard CB]{
    \centering
    \includegraphics[width=0.58\linewidth]{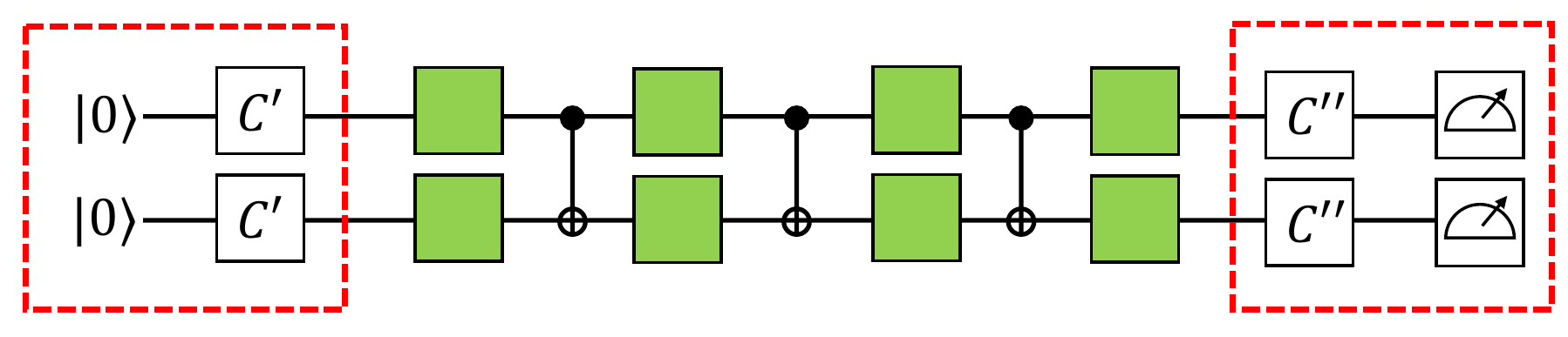}
    }\\
    \subfloat[CB with interleaved gates]{
    \centering
    \includegraphics[width=0.7\linewidth]{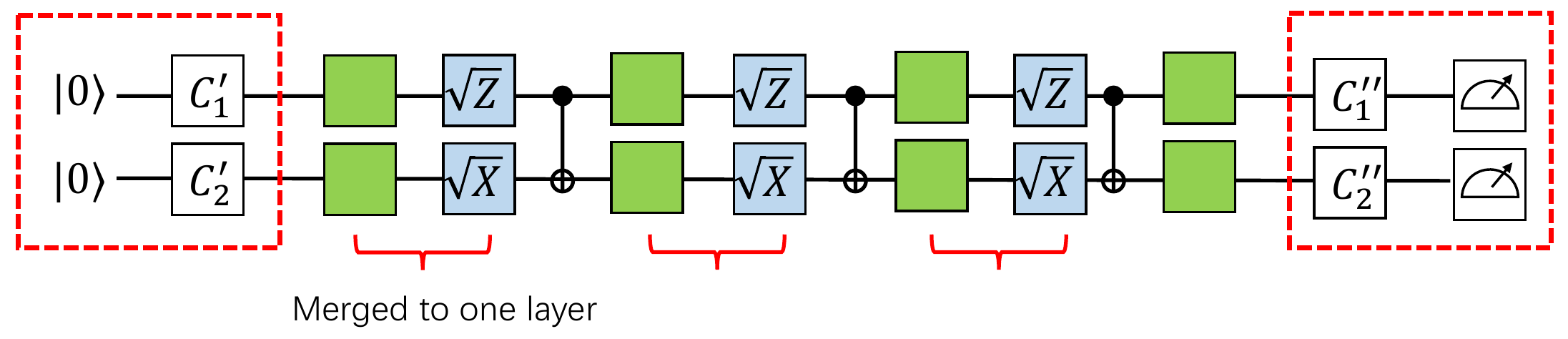}
    }
    \caption{Cycle benchmarking for learning the Pauli noise channel of a CNOT gate. (a) Standard CB circuits, where CNOT gates are interleaved by random Pauli gates (green boxes), with initial stabilizer states and Pauli basis measurements (red boxes). (b) CB circuits with additional interleaved single qubit Clifford gates (blue boxes).}
    \label{fig:main_cb}
\end{figure}

The convenience of working with Pauli fidelities is further demonstrated by the fact that some Pauli fidelities can be directly learned by cycle benchmarking, even with noisy state preparation and measurement. For example, consider the CNOT gate which maps the Pauli operator $IX$ to itself. Fig.~\ref{fig:main_cb} (a) shows the cycle benchmarking circuit. Imagine that we put the Pauli operator $IX$ after the left red box and evolve it with the circuit, then the evolved operator (before the right red box) equals $\lambda_{IX}^3\cdot IX$, up to a $\pm$ sign (which comes from the random Pauli gates and can always be accounted for during post-processing).
Here we use the convention that the noise channel happens before each CNOT gate. In experiments, we prepare a $+1$ eigenstate of $IX$ (such as $\ket{+}\ket{+}$), measure the expectation value of $IX$ at the end, and average over random Pauli twirling sequences. These SPAM operations are noisy and are represented as the red boxes. It is shown~\cite[Theorem~1 in Supplementary Information]{erhard2019characterizing} that the measured expectation value equals
\begin{equation}
    \E \expval{IX}=A_{IX}\cdot \lambda_{IX}^d
\end{equation}
where the expectation is over random Pauli twirling gates and randomness of quantum measurement, and $A_{IX}$ depends on SPAM noise but is independent of circuit depth $d$. From this $\lambda_{IX}$ can be learned by estimating the observable $IX$ at several different depths and perform a curve fitting.

The Pauli operator $IX$ is special as it is invariant under CNOT. Consider another example: CNOT maps $XZ$ to $YY$ and vice versa. Consider Fig.~\ref{fig:main_cb} (b) where we insert additional layers of single-qubit Clifford gates $\sqrt{Z}\otimes\sqrt{X}$ that also maps $XZ$ to $YY$ and vice versa (up to a minus sign that can always be accounted for during post-processing). After $XZ$ picks up a coefficient $\lambda_{XZ}$ in front of the CNOT gate, it gets mapped to $\lambda_{XZ}\cdot YY$ by CNOT but then rotated back to $\lambda_{XZ}\cdot XZ$ by $\sqrt{Z}\otimes\sqrt{X}$. Following the same argument we conclude that both $\lambda_{XZ}$ and $\lambda_{YY}$ are learnable. For simplicity here we make an assumption that single qubit gates are noiseless, motivated by the fact that single qubit gates are 1-2 magnitudes less noisy than 2-qubit gates on today's quantum hardware~\cite{ibmquantum}. In practice, it is a standard assumption to model noise on single-qubit gates as gate-independent (\textit{e.g.}~\cite[Sec. II A]{Ferracin2022Efficiently}), and our noise characterization result can be interpreted as the noise channel induced by a dressed cycle which consists of a CNOT gate and two single-qubit gates~\cite{wallman2016noise}.

The main challenge comes with the next example: CNOT maps $IZ$ to $ZZ$ and vice versa. By directly applying cycle benchmarking as in Fig.~\ref{fig:main_cb} (a) (with even depth $d$) we obtain
\begin{equation}
    \E\expval{IZ}=A_{IZ}\cdot \lambda_{IZ}\lambda_{ZZ}\lambda_{IZ}\lambda_{ZZ}\cdots = A_{IZ}\left(\lambda_{IZ}\lambda_{ZZ}\right)^{d/2},
\end{equation}
and curve fitting gives $\sqrt{\lambda_{IZ}\lambda_{ZZ}}$ (similar results have been obtained in~\cite{erhard2019characterizing,hashim2020randomized,berg2022probabilistic,Ferracin2022Efficiently}). To learn $\lambda_{IZ}$, we may consider applying the same technique in Fig.~\ref{fig:main_cb} (b). However, the problem is that once $IZ$ gets mapped to $ZZ$, it cannot be rotated back to $IZ$ because $I$ is invariant under single qubit unitary gates. The main difference between this example and previous examples is that here the \emph{Pauli weight pattern} (an $n$-bit binary string with 0 indicating identity and 1 indicating non-identity) changes from 01 to 11, thus making the single qubit rotation tool inapplicable.

In fact we can go on to prove that $\lambda_{IZ}$ (as well as $\lambda_{ZZ}$) is unlearnable. Here unlearnable means that there exists two noise models such that the parameter $\lambda_{IZ}$ is different, but the two noise models are indistinguishable by any quantum experiment, meaning that any quantum experiment generates exactly the same output statistics with the two noise models. The result also generalizes to arbitrary $n$-qubit Clifford gates.

\begin{theorem}\label{thm:mainpaulifidelity}
Given an $n$-qubit Clifford gate $\mc G$ and an $n$-qubit Pauli operator $P_a$,
the Pauli fidelity $\lambda_a$ of the noise channel attached to $\mc G$ is learnable if and only if $\pt(\mc G(P_a))= \pt(P_a)$. Here $\pt$ denotes the Pauli weight pattern.
\end{theorem}

The ``if'' part follows directly from cycle benchmarking as discussed above. For the ``only if'' part, when $\pt(\mc G(P_a))\neq \pt(P_a)$, we construct a gauge transformation to prove the unlearnability of $\lambda_a$, following ideas from gate set tomography~\cite{Merkel2013Self-consistent,Blume-Kohout2013Robust,nielsen2021gate}. A gauge transformation is an invertible linear map $\mc M$ that converts a noise model (initial states $\rho_i$, POVM operators $E_j$, noisy gates $G_k$) to a new noise model as
\begin{equation}
    \rho_i\mapsto \mc M(\rho_i),\quad E_j\mapsto (\mc M^{-1})^\dagger (E_j),\quad G_k\mapsto \mc M\circ G_k \circ\mc M^{-1},
\end{equation}
with the constraint that the new noise model is physical. Note that the old and new noise models are indistinguishable by definition. To construct such a gauge transformation, as $\pt(\mc G(P_a))\neq \pt(P_a)$, there exists a bit on which the two Pauli weight patterns differ. We then define $\mc M$ as a single-qubit depolarizing noise channel on the corresponding qubit. In this way we can show that the old and new noise models assign different values to $\lambda_a$, which means $\lambda_a$ is unlearnable. This proof naturally implies that using other noisy gates from the gate set (that are subject to different unknown noise channels) does not change the learnability of Pauli fidelities. 
More details of the proof are given in Supplementary Section II B.
As a side remark, it is known that under the stronger assumption of gate-independent noise (where different multi-qubit gates are assumed to have the same noise channel), the noise channel is fully learnable~\cite{kimmel2014robust,helsen2021estimating,huang2022foundations}.

Theorem~\ref{thm:mainpaulifidelity} provides a simple condition for determining the learnability of individual Pauli fidelities, but it is not sufficient for characterizing the learnability of joint functions of different Pauli fidelities. In the CNOT example, we know that both $\lambda_{IZ}$ and $\lambda_{ZZ}$ are unlearnable, but we also know that their product $\lambda_{IZ}\lambda_{ZZ}$ is learnable. This means that there is only one unlearnable degree of freedom in the two parameters $\{\lambda_{IZ},\lambda_{ZZ}\}$. In the following we show how to determine learnable and unlearnable degrees of freedom of Pauli noise, and also generalize the discussion from a single gate to a gate set.

We start by defining learnable information. Consider a Clifford gate set with $m$ gates, where we model each gate as an $n$-qubit gate associated with an $n$-qubit Pauli noise channel. This model is applicable to both individual gates (\textit{e.g.} a 2-qubit system where each 2-qubit gate is implemented by a different physical process and subject to a different noise channel) as well as parallel applications of gates (\textit{e.g.} an $n$-qubit system where each ``gate'' in the gate set is implemented by a layer of 2-qubit gates; the $n$-qubit noise channel models the crosstalk among the 2-qubit gates). The goal is to characterize the learnable degrees of freedom among the $m\cdot 4^n$ parameters.  

Recall that the output of cycle benchmarking is a product of Pauli fidelities (including SPAM noise). We further show that without loss of generality this is the only type of information that we need to obtain from quantum experiments for the purpose of noise learning. This is because in general the output probability of any quantum experiment can be expressed as a sum of products of Pauli fidelities, and each individual product can be learned by cycle benchmarking (Supplementary Section IV). We therefore consider learning functions of the noise model that can be expressed as a product of Pauli fidelities (also see below Eq.~\eqref{eq:mainapproximation} for a related discussion). This can be reduced to considering functions of the form $f=\sum_{a,\mc G}v_{a}^{\mc G}\cdot l_a^{\mc G}$, where $l_a^{\mc G}:=\log \lambda_a^{\mc G}$ is the log Pauli fidelity, $v_{a}^{\mc G}\in\mathbb{R}$, and the superscript $\mc G$ denotes the corresponding Clifford gate. In the CNOT example $l_{IZ}+l_{ZZ}$ is a learnable function. The idea of learning log Pauli fidelities in benchmarking has also been considered in~\cite{flammia2021averaged,nielsen2022first}. The advantage of considering log Pauli fidelities here is that the set of all learnable functions $f$ forms a vector space. Therefore to characterize all independent learnable degrees of freedom, we only need to determine a basis of the vector space.

\begin{figure}[t]
    \centering
    \includegraphics[width=1.0\linewidth]{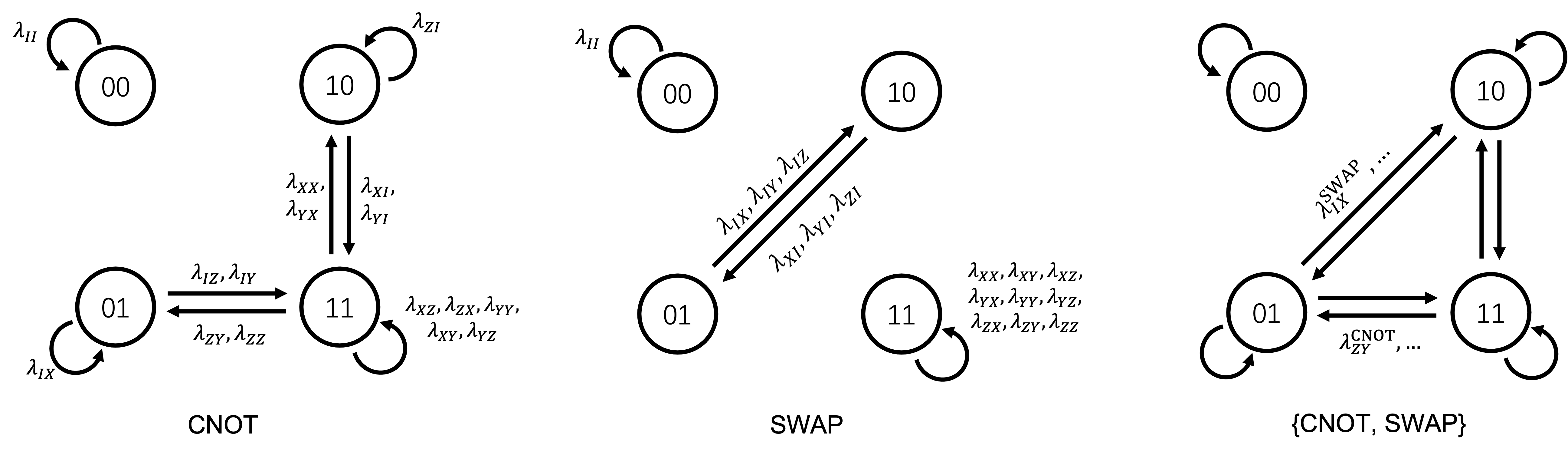}
    \caption{Pattern transfer graph of CNOT, SWAP, and a gate set consisting of CNOT and SWAP. Here, multiple edges are represented by a single edge with multiple labels. 
    The labels on the first two graphs are gate dependent, though we omit the superscripts of CNOT or SWAP.
    The labels on the last graph are a combination of the first two graphs and are omitted for clarity.}
    \label{fig:main_patterntransfer}
\end{figure}

Recall that the reason that $l_{IZ}+l_{ZZ}$ is learnable in the CNOT example is because the path of Pauli operator in the cycle benchmarking circuit forms a cycle $IZ\to ZZ\to IZ\to\cdots$, and the product of Pauli fidelities along the cycle ($\lambda_{IZ}\lambda_{ZZ}$) can be learned via curve fitting. In general, as we can also insert single qubit Clifford gates in between, we do not need to differentiate between $X,Y,Z$. We therefore consider the \emph{pattern transfer graph} associated with a Clifford gate set where vertices corresponds to binary Pauli weight patterns and each edge is labeled by the Pauli fidelity of the incoming Pauli operator. The graph has $2^n$ vertices and $m\cdot 4^n$ directed edges. They can also be merged to form the pattern transfer graph of the gate set $\{\mathrm{CNOT},\mathrm{SWAP}\}$. Fig.~\ref{fig:main_patterntransfer} shows the pattern transfer graph of CNOT, SWAP, and the gate set of $\{\text{CNOT}, \text{SWAP}\}$. Consider an arbitrary cycle in the pattern transfer graph $C=(e_1,\dots,e_k)$ where each edge $e_i$ is associated with some Pauli fidelity $\lambda_i$. Following Fig.~\ref{fig:main_cb} (b), a cycle benchmarking circuit can be constructed which learns the product of the Pauli fidelites along the cycle, or equivalently the function $f_C:=\sum_{e_i\in C}\log \lambda_i$ can be learned. This implies that the set of functions defined by linear combination of cycles $\{\sum_{C\in\text{cycles}}\alpha_C f_C:\alpha_C\in\mathbb{R}\}$ are learnable. In the following we show that this in fact corresponds to all learnable information about Pauli noise.

We label the edges of the pattern transfer graph as $e_1,\dots,e_M$ where $M=m\cdot 4^n$ and each edge $e_i$ is a variable that represents some log Pauli fidelity. The goal is to characterize the learnability of linear functions of the edge variables $f=\sum_{i=1}^M v_i e_i$, $v_i\in\mathbb{R}$. The set of linear functions can be equivalently understood as a vector space of dimension $M$, called the \emph{edge space} of the graph, where $f$ corresponds to a vector $(v_1,\dots,v_M)$ and we think of $e_1,\dots,e_M$ as the standard basis. Following the above discussion, the \emph{cycle space} of the graph is defined as $\vspan\{\sum_{e\in C}e:C\text{ is a cycle}\}$, which is a subspace of edge space. We also define another subspace, the \emph{cut space}, as $\vspan\{\sum_{e\in C}(-1)^{e\text{ from }V_1\text{ to }V_2}e:C\text{ is a cut between a partition of vertices }V_1,V_2\}$. It is known that the edge space is the orthogonal direct sum of cycle space and cut space for any graph~\cite{bollobas1998modern}. Interestingly, we show that the complementarity between cycle and cut space happens to be the dividing line that determines the learnability of Pauli noise. 

\begin{theorem}\label{thm:mainpaulilearnability}
  The vector space of learnable functions of the Pauli noise channels associated with an $n$-qubit Clifford gate set is equivalent to the cycle space of the pattern transfer graph. In other words,
  \begin{equation}
      \begin{aligned}
       \text{All information}\quad &\equiv \quad \text{Edge space},\\
       \text{Learnable information}\quad &\equiv \quad \text{Cycle space},\\
       \text{Unlearnable information}\quad &\equiv \quad \text{Cut space}.\\
      \end{aligned}
  \end{equation}
This implies that the number of unlearnable degrees of freedom equals $2^n - c$, where $c$ is the number of connected components of the pattern transfer graph.
\end{theorem}

The learnability of cycle space follows from cycle benchmarking as discussed above. To prove the unlearnability of cut space, we use a similar argument as in Theorem~\ref{thm:mainpaulifidelity} and show that a gauge transformation can be constructed for each cut in the pattern transfer graph. By linearity, this implies that any vector in the cut space corresponds to a gauge transformation. By definition, a learnable function must be orthogonal to all such vectors and thus orthogonal to the entire cut space. More details of the proof are given in Supplementary Section II C.

It is a well-known fact in graph theory that the cycle space of a directed graph $G=(V,E)$ has dimension $|E|-|V|+c$ while the cut space has dimension $|V|-c$, where $c\geq 1$ is the number of connected components in $G$~\cite{bollobas1998modern} 
(a (weakly) connected component is a maximal subgraph in which every vertex is reachable from every other vertex via an undirected path).
Theorem~\ref{thm:mainpaulilearnability} implies that among the $m\cdot 4^n$ degrees of freedom of the Pauli noise associated with a Clifford gate set, there are $2^n -c$ unlearnable degrees of freedom. This shows that while the number of unlearnable degrees of freedom can be exponentially large, they only occupy an exponentially small fraction of the entire space. In addition, a cycle and cut basis can be efficiently determined for a given graph, though in our case this takes exponential time because the pattern transfer graph itself is exponentially large. However, computing the cycle/cut basis is not the bottleneck as the information to be learned also grows exponentially with the number of qubits.
For small system sizes such as 2-qubit Clifford gates, we can write down a cycle basis as shown in Table~\ref{tab:main:CNOT_full} (a) for the CNOT and SWAP gates, which represents all learnable information about these gates. The CNOT gate has 2 unlearnable degrees of freedom while the SWAP gate has 1 unlearnable degree of freedom. As the pattern transfer graph has at least 2 connected components, we conclude that the Pauli noise channel of a 2-qubit Clifford gate has at most 2 unlearnable degrees of freedom. Note that when treating $\{\mathrm{CNOT},\mathrm{SWAP}\}$ together as a gate set, there are only 2 unlearnable degrees of freedom according to Theorem~\ref{thm:mainpaulilearnability} instead of $2+1=3$, because there is one additional learnable degree of freedom (such as $l_{IZ}^{\mathrm{CNOT}}+l_{XX}^{\mathrm{CNOT}}+l_{XI}^{\mathrm{SWAP}}$) that is a joint function of the two gates.

\begin{table}[t]
    \centering
    \begin{tabular}{|c|c|c|}
        \hline
        Gate & CNOT & SWAP \\
        \hline
        \makecell{(a) Cycle basis}
          &\makecell{ $l_{II},l_{ZI},l_{IX},l_{ZX},
        l_{XZ},l_{YY},l_{XY},l_{YZ},$ \\ $l_{IZ}+l_{ZZ},l_{IY}+l_{ZY},l_{IZ}+l_{ZY},$\\$l_{XI}+l_{XX},l_{YI}+l_{YX},l_{XI}+l_{YX}$ }  & \makecell{$l_{II},l_{XX},l_{XY},l_{XZ},l_{YX},l_{YY},l_{YZ},l_{ZX},l_{ZY},$\\$l_{ZZ},l_{IX}+l_{XI},l_{IY}+l_{YI},l_{IZ}+l_{ZI},$\\$l_{XI}+l_{IY},l_{XI}+l_{IZ}$}
        \\ \hline
        \makecell{(b) Learnable\\ Pauli fidelities}
         &\makecell{ $\lambda_{II},\lambda_{ZI},\lambda_{IX},\lambda_{ZX},
        \lambda_{XZ},\lambda_{YY},\lambda_{XY},\lambda_{YZ},$ 
        \\$\lambda_{IZ}\cdot\lambda_{ZZ},\lambda_{IY}\cdot\lambda_{ZY},\lambda_{IZ}\cdot\lambda_{ZY},$\\$\lambda_{XI}\cdot\lambda_{XX},\lambda_{YI}\cdot\lambda_{YX},\lambda_{XI}\cdot\lambda_{YX}$}  &\makecell{$\lambda_{II},\lambda_{XX},\lambda_{XY},\lambda_{XZ},\lambda_{YX},\lambda_{YY},\lambda_{YZ},\lambda_{ZX},\lambda_{ZY},$\\$\lambda_{ZZ},\lambda_{IX}\cdot\lambda_{XI},\lambda_{IY}\cdot\lambda_{YI},\lambda_{IZ}\cdot\lambda_{ZI},$\\$\lambda_{XI}\cdot\lambda_{IY},\lambda_{XI}\cdot\lambda_{IZ}$} 
        \\ \hline
        \makecell{(c) Learnable\\ Pauli errors}
         &\makecell{ $p_{II},p_{ZI},p_{IX},p_{ZX},
        p_{XZ},p_{YY},p_{XY},p_{YZ},$ 
        \\$p_{IZ}+p_{ZZ},p_{IY}+p_{ZY},p_{IZ}+p_{ZY},$\\$p_{XI}+p_{XX},p_{YI}+p_{YX},p_{XI}+p_{YX}$}  &\makecell{$p_{II},p_{XX},p_{XY},p_{XZ},p_{YX},p_{YY},p_{YZ},p_{ZX},p_{ZY},$\\$p_{ZZ},p_{IX}+p_{XI},p_{IY}+p_{YI},p_{IZ}+p_{ZI},$\\$p_{XI}+p_{IY},p_{XI}+p_{IZ}$} \\
        \hline
        \makecell{(d) Unlearnable\\  degrees of freedom}
         & $\lambda_{XI},\lambda_{IZ}$ & $\lambda_{XI}$ \\
        \hline
    \end{tabular}
    \caption{A complete basis for the learnable linear functions of log Pauli fidelities and Pauli error rates for a single CNOT/SWAP gate.
    }
    \label{tab:main:CNOT_full}
\end{table}

Finally, the learnability of Pauli errors can be determined by the learnability of Pauli fidelities according to the Walsh-Hadamard transform $p_a = \frac{1}{4^n}\sum_{b\in{{\sf P}^n}}\lambda_b(-1)^\expval{a,b}$. An issue here is that Pauli errors are linear functions of $\{\lambda_b\}$ instead of $\{\log \lambda_b\}$. Here we make a standard assumption in the literature~\cite{erhard2019characterizing,flammia2020efficient} that the total Pauli error is sufficiently small. In this case all individual Pauli errors are close to 0 while all individual Pauli fidelities are close to 1. Therefore the Pauli errors can be estimated via
\begin{equation}\label{eq:mainapproximation}
    p_a = \frac{1}{4^n}\sum_{b\in{{\sf P}^n}}\lambda_b(-1)^\expval{a,b}\approx\frac{1}{4^n}\sum_{b\in{{\sf P}^n}}(-1)^\expval{a,b}\left(1+\log\lambda_b\right),
\end{equation}
which means that their learnability can be determined by Theorem~\ref{thm:mainpaulilearnability}. In fact it has been suggested~\cite{nielsen2022first} that any function of Pauli fidelities can be estimated in this way (as a linear function of log Pauli fidelities) up to a first-order approximation, which means that the learnability of any function of Pauli fidelities can be determined by Theorem~\ref{thm:mainpaulilearnability}. In Table~\ref{tab:main:CNOT_full} (c) we show the learnable Pauli errors for CNOT and SWAP, where ``learnable'' is in an approximate sense up to Eq.~\eqref{eq:mainapproximation}. 
Interestingly, for these two gates, the learnable functions of Pauli errors have the same form as the cycle basis, \textit{i.e.} the cycle space is invariant under Walsh-Hadamard transform. 
We calculate the learnable Pauli errors for up to 4-qubit random Clifford gates and this seems to be true in general.
We leave a rigorous investigation into this phenomenon for future work.

\subsection{Experiments on IBM Quantum hardware}
We demonstrate our theory on IBM quantum hardware~\cite{ibmquantum} using a minimal example -- characterizing the noise channel of a CNOT gate. In our experiments both the gate noise and SPAM noise are twirled into Pauli noise using randomized compiling. In the following we show how to extract all learnable information of Pauli noise SPAM-robustly, and also attempt to estimate the unlearnable degrees of freedom by making additional assumptions.

First, we conduct two types of cycle benchmarking (CB) experiments, the standard CB and CB with interleaving single-qubit gates (called \emph{interleaved CB}), as shown in Fig.~\ref{fig:main_cb}.
The results are shown in Fig.~\ref{fig:main_exp_cbraw}. Here a set of two Pauli labels in the $x$-axis (\textit{e.g.}, $\{IZ,ZZ\}$) corresponds to the geometric mean of the Pauli fidelity (\textit{e.g.}, $\sqrt{\lambda_{IZ}\lambda_{ZZ}}$).
Comparing to Table~\ref{tab:main:CNOT_full}, we see that all learnable information of Pauli fidelities (including learnable individual and 2-product) are successfully extracted.
Also note from Fig.~\ref{fig:main_exp_cbraw} that the two types of CB experiments give consistent estimates, in terms of both the process fidelity and individual Pauli fidelities (\textit{e.g.}, $\sqrt{\lambda_{XZ}\lambda_{YY}}$ estimated from standard CB is consistent with $\lambda_{XZ}$ and $\lambda_{YY}$ from interleaved CB).

\begin{figure}[t]
    \centering
    \includegraphics[width=\linewidth]{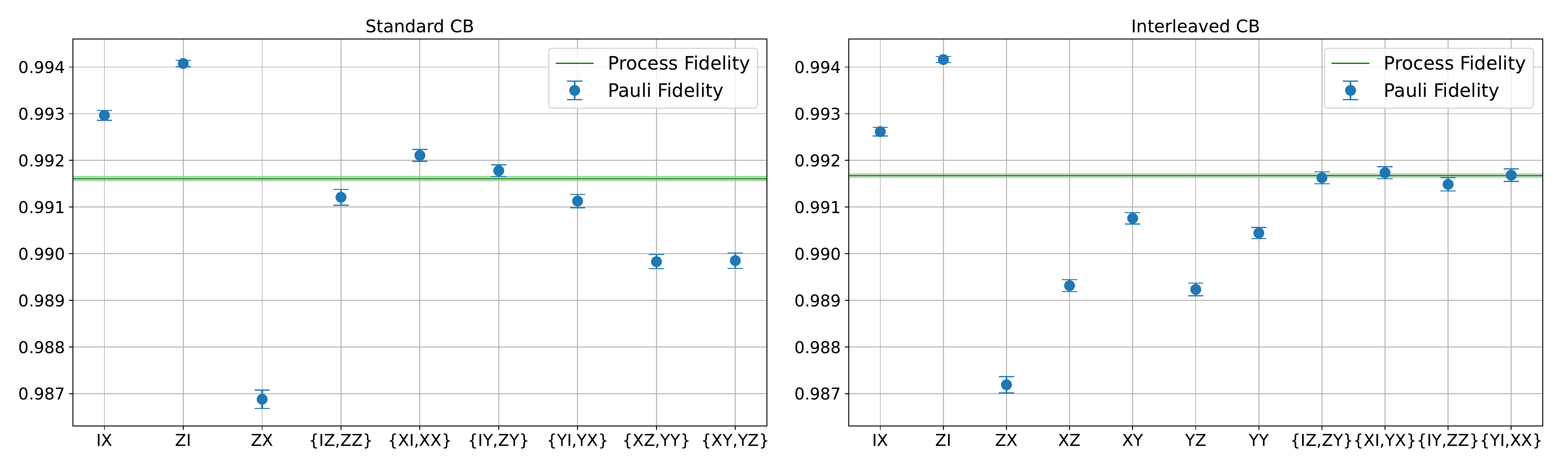}
    \caption{Estimates of Pauli fidelities of IBM's CNOT gate via standard CB (left) and CB with interleaved gates (right), using circuits shown in Fig.~\ref{fig:main_cb}. Data are collected from \texttt{ibmq\_montreal} on 2022-03-23. Each Pauli fidelity is fitted using seven different circuit depths $L=[2,2^2,...,2^7]$. For each depth $C=60$ random circuits and $1000$ shots of measurements are used. Throughout this paper, the error bar represents the standard error.
    }
    \label{fig:main_exp_cbraw}
\end{figure}

We have shown that all 13 learnable degrees of freedom (excluding the trivial $\lambda_{II}=1$) are extracted in Fig.~\ref{fig:main_exp_cbraw} by comparing with Table~\ref{tab:main:CNOT_full}, and there remain 2 unlearnable degrees of freedom. We can bound the feasible region of the 2 unlearnable degrees of freedom using physical constraints, \textit{i.e.}, the reconstructed Pauli noise channel must be completely positive. 
This is equivalent to requiring $p_a\ge 0$ for all Pauli error rates $p_a$. 
We choose $\lambda_{XX}$ and $\lambda_{ZZ}$ as a representation of the unlearnable degrees of freedom, and plot the calculated feasible region in Fig.~\ref{fig:main_exp_cbfeasible} (a), which happens to be a rectangular area. We also calculate the feasible region for each unlearnable Pauli fidelity and Pauli error rate, which are presented in Fig.~\ref{fig:main_exp_cbfeasible} (b), (c). In particular, we choose two extreme points (blue and green dots in Fig.~\ref{fig:main_exp_cbfeasible} (a)) in the feasible region and plot the corresponding noise model in Fig.~\ref{fig:main_exp_cbfeasible} (b), (c). Note that the (approximately) learnable Pauli error rates (on the left of the red vertical dashed line) are nearly invariant under change of gauge degrees of freedom, but they can be estimated to be negative due to statistical fluctuation. Thus, when we calculate the physical constraints, we only require those unlearnable Pauli error rates (on the right of the red vertical dashed line) to be non-negative.

\begin{figure}[t]
    \centering
    \subfloat[feasible region]{
    \centering
    \includegraphics[width=0.35\linewidth]{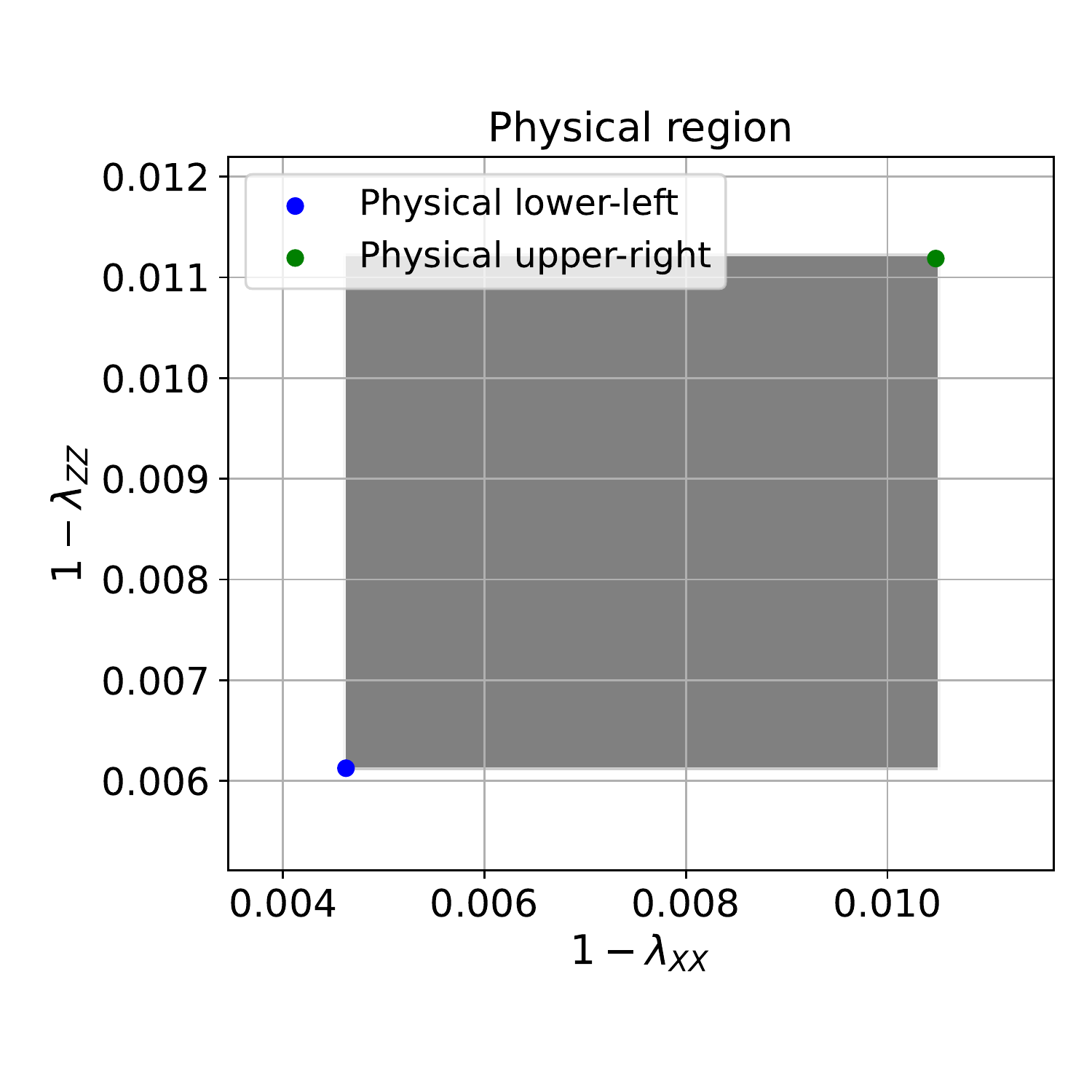}
    }\\
    \subfloat[Pauli fidelities]{
    \centering
    \includegraphics[width=0.5\linewidth]{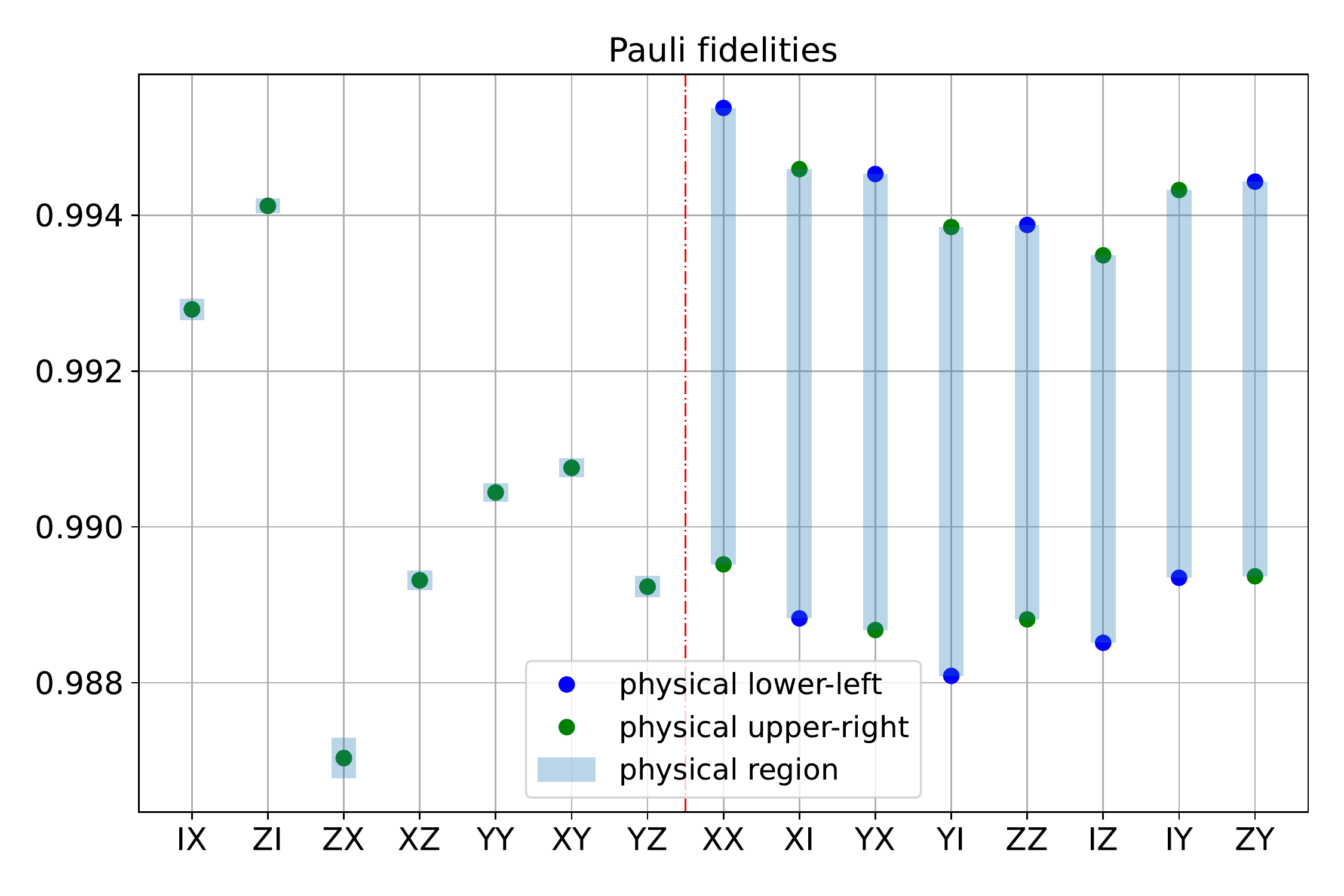}
    }
    \subfloat[Pauli errors]{
    \centering
    \includegraphics[width=0.5\linewidth]{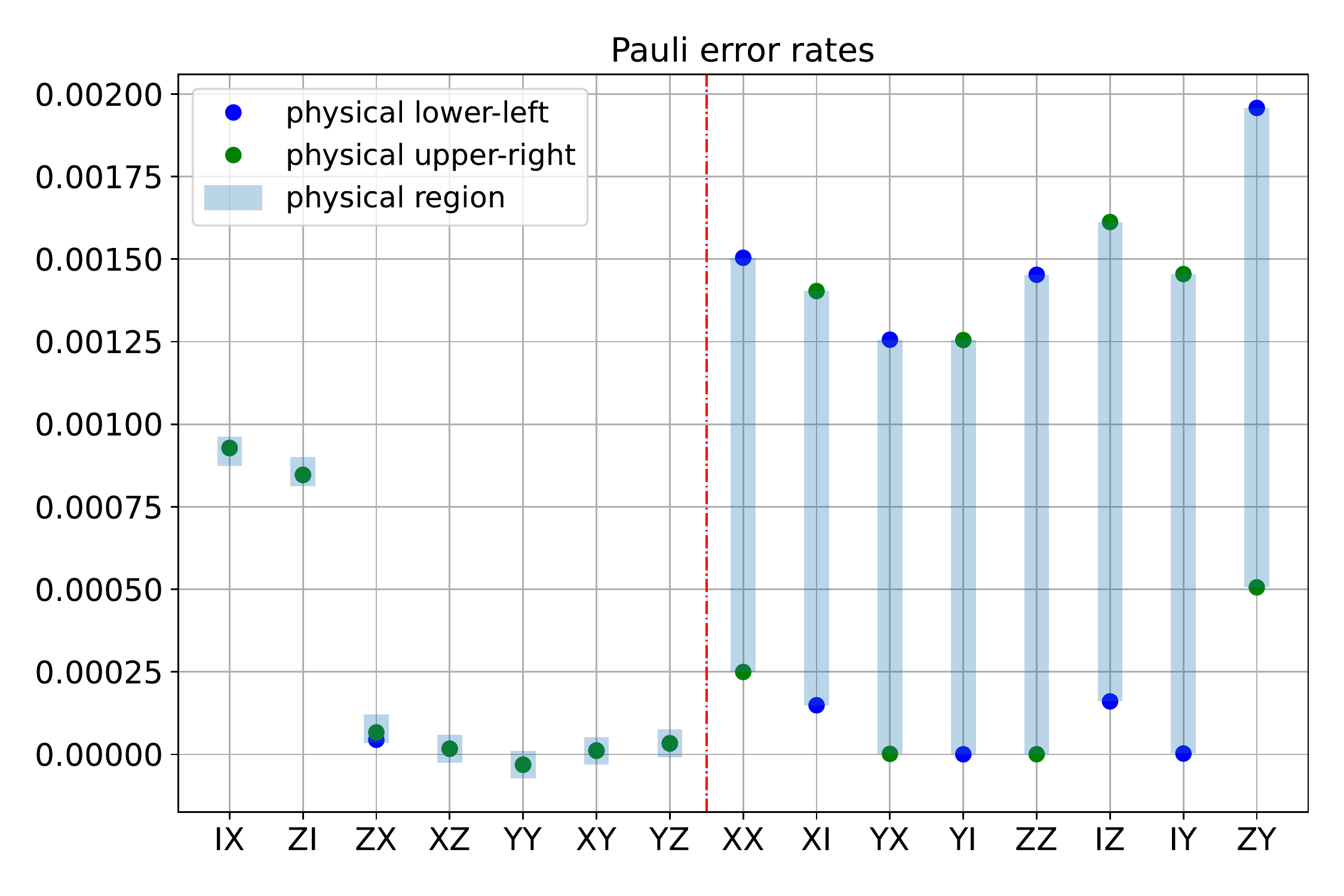}
    }
    \caption{Feasible region of the learned Pauli noise model, using data from Fig.~\ref{fig:main_exp_cbraw}. (a) Feasible region of the unlearnable degrees of freedom in terms of $\lambda_{XX}$ and $\lambda_{ZZ}$. (b) Feasible region of individual Pauli fidelities. (c) Feasible region of individual Pauli errors.}
    \label{fig:main_exp_cbfeasible}
\end{figure}

Next, we explore an approach to estimate the unlearnable information with additional assumptions.
Suppose that one can prepare $\ket{0}^{\otimes n}$ perfectly.
Since we assume noiseless single-qubit gates, this means we can prepare a set of perfect tomographically complete states $\{\ket{0/1},\ket{\pm},\ket{\pm i}\}$.
In this case, all the unlearnable degrees of freedom become learnable, as one can first perform a measurement device tomography, and then directly estimate the process matrix of a noisy gate with measurement error mitigated~\cite{Maciejewski2020mitigationofreadout}. 
Following this general idea, we propose a variant of cycle benchmarking for Pauli noise characterization, which we call \emph{intercept CB} as it uses the information of intercept in a standard cycle benchmarking protocol. Given an $n$-qubit Clifford gate $\mc G$, let $m_0$ be the smallest positive integer such that $\mc G^{m_0} = \mc I$.
For any Pauli fidelity $\lambda_a$ (regardless of whether learnable or not according to Theorem~\ref{thm:mainpaulifidelity}), consider the following two CB experiments using the standard circuit as in Fig.~\ref{fig:main_cb} (a). First, prepare an eigenstate of $P_a$, run CB with depth $l m_0+1$ for some non-negative integer $l$, and estimate the expectation value of $P_b\coleq\mc G( P_a )$. The result equals 
\begin{equation}
    \E\expval{P_b}_{l m_0+1}=\lambda^S_{P_a}\lambda^M_{P_b}\lambda_{a}\left(\prod_{k=1}^{m_0}\lambda_{\mc G^k(P_a)}\right)^l,
\end{equation}
where $\lambda_{P_{a/b}}^{S/M}$ is the Pauli fidelity of the state preparation and measurement noise channel, respectively (earlier we have absorbed these two coefficients into a single coefficient $A$ for simplicity). Second, prepare an eigenstate of $P_b$, run CB with depth $l m_0$, and estimate the expectation value of $P_b$. The result equals
\begin{equation}
    \E\expval{P_b}_{l m_0}=\lambda^S_{P_b}\lambda^M_{P_b}\left(\prod_{k=1}^{m_0}\lambda_{\mc G^k(P_a)}\right)^l.
\end{equation}
By fitting both $\E\expval{P_b}_{l m_0+1}$ and $\E\expval{P_b}_{l m_0}$ as exponential decays in $l$, extracting the intercepts (function values at $l=0$), and taking the ratio,
we obtain an estimator $\widehat{\lambda}^{\text{ICB}}_a$
that is asymptotically unbiased to $\lambda_{a}\cdot{\lambda^{S}_{P_a}}/{\lambda^{S}_{P_b}}$.
This estimator is robust against measurement noise. Note that $\lambda^S_{P_a}=\lambda^S_{P_b}=1$ if we assume perfect initial state preparation, and in this case the above shows that $\lambda_a$ is learnable, and thus the entire Pauli noise channel is learnable.
We note that, instead of fitting an exponential decay in $l$, one could in principle just take $l=0$ and estimate the ratio of $\E\expval{P_b}_{0}$ and $\E\expval{P_b}_{1}$, which also yields a consistent estimate for $\lambda_a\cdot\lambda^S_{P_a}/\lambda^S_{P_b}$. 
If one has already obtained all the learnable information from previous experiments, this could be a more efficient approach. 
However, if one has not done those experiments, the intercept CB with multiple depths can estimate the intercept (unlearnable information) and slope (learnable information) simultaneously, which is more sample efficient. 

We numerically simulate intercept CB for characterizing the CNOT gate under different state preparation (SP) and measurement (M) noise. As shown in Fig.~\ref{fig:main_sim_intercept}, this method yields relatively precise estimate when there is only measurement noise even if the noise is orders of magnitude stronger than the gate noise, but will have large deviation from the true noise model even under small state preparation noise. We refer the reader to Supplementary Section III for more details about the numerical simulation.

Finally, we experimentally implement intercept CB to estimate $\lambda_{XX}$ and $\lambda_{ZZ}$, which are the two unlearnable degrees of freedom of CNOT, allowing us to determine all the Pauli fidelities and Pauli error rates. 
One challenge in interpreting the results is that we do not know in general whether the low SP noise assumption holds, therefore it is unclear if the learned results should be trusted. 
However, for the estimate to be correct, it should at least lie in the physically feasible region we obtained earlier in Fig.~\ref{fig:main_exp_cbfeasible}.
In Fig.~\ref{fig:main_exp_intercept}, we present our experimental results of intercept CB. It turns out that certain Pauli fidelities are far away from the physical region by several standard deviations. 
This gives strong evidence that the low SP noise assumption was \emph{not} true on the platform we used.

The data collected here can further be used to give a lower bound for the SP noise. Suppose we obtain the physical region of $\lambda_a$ to be $[\widehat{\lambda}_{a,\mr{min}},\widehat{\lambda}_{a,\mr{max}}]$. Combining with the expression of intercept CB, we have
\begin{equation}
    {\widehat{\lambda}^\mr{ICB}_a}/{\widehat{\lambda}_{a,\mr{max}}}\le{\lambda^S_{P_a}}/{\lambda^S_{P_b}}\le{\widehat{\lambda}^\mr{ICB}_a}/{\widehat{\lambda}_{a,\mr{min}}}.
\end{equation}
Applying this to the data of $IZ$ and $ZZ$ in Fig.~\ref{fig:main_exp_intercept} (a), we have $\lambda^S_{IZ}/\lambda^S_{ZZ} \le 0.9879(23)$.
If we make a physical assumption that the state preparation noise is a random bit-flip during the qubit initialization, one can conclude the bit-flip rate on the first qubit is lower bounded by $0.61(12)\%$.
One can in principle bound the bit-flip rate on the second qubit by looking at $\lambda^S_{XX}/\lambda^S_{XI}$. Unfortunately, our estimate of $\lambda^S_{XX}$ from intercept CB falls in the physical region within one standard deviation, so there is no nontrivial lower bound. One could expect obtaining a useful lower bound by looking at a CNOT gate with reversed control and target.
The lower bound of SP noise obtained here is completely independent of the measurement noise and does not suffer from the issue of gauge freedom~\cite{nielsen2021gate}, as long as all of our noise assumptions are valid, \textit{i.e.}, there is no significant contribution from time non-stationary, non-Markovian, or single-qubit gate-dependent noise.

\begin{figure}
    \centering
    \includegraphics[width=0.6\linewidth]{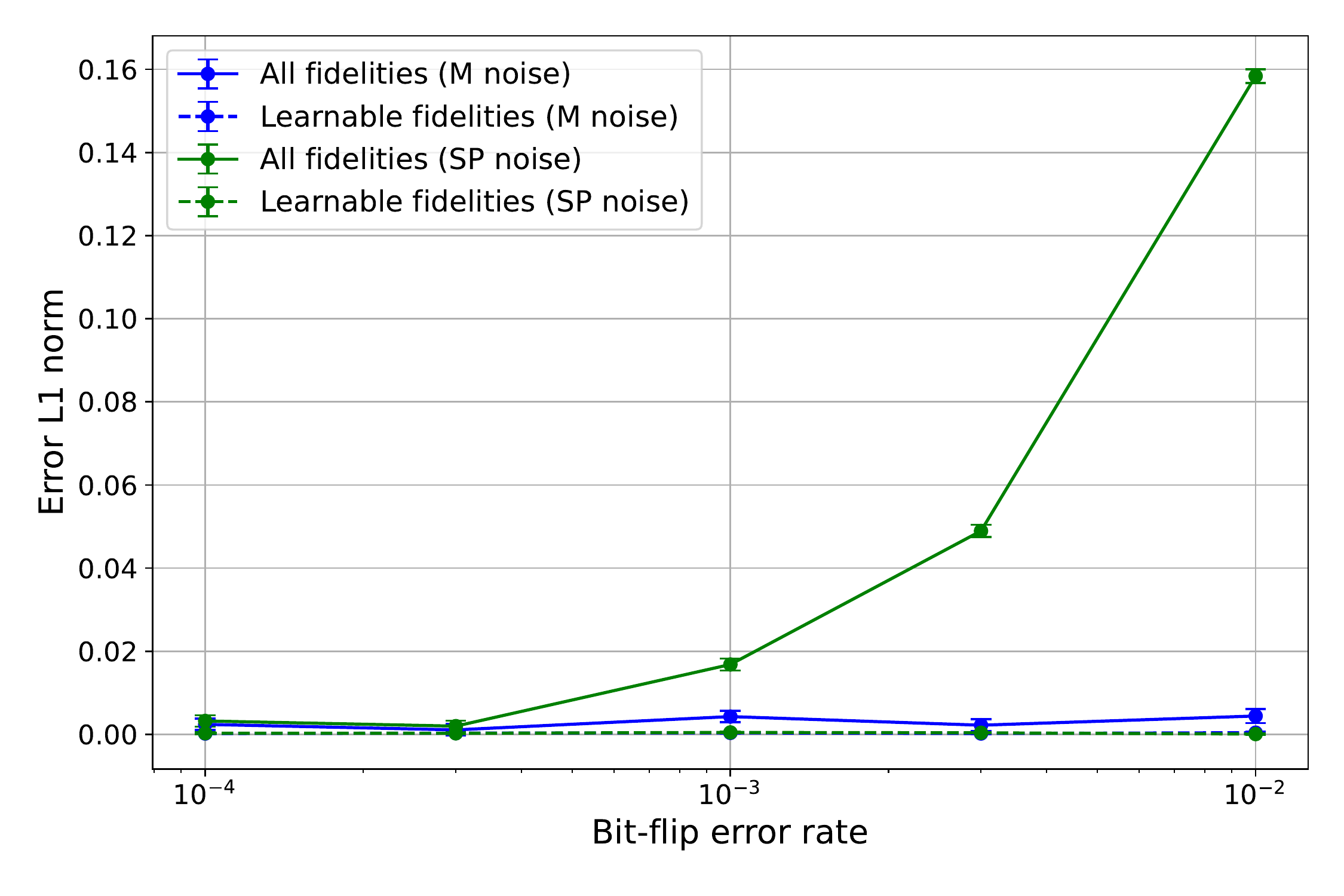}
    \caption{Simulation of intercept CB on CNOT under different SPAM noise rate. The simulated noise channel is a $2$-qubit amplitude damping channel with effective noise rate $5\%$, and SPAM noise are modeled as bit-flip errors. For the blue (green) lines, we introduce random bit-flip errors to the measurement (state preparation). The solid lines show the $l_1$-distance of the estimated Pauli fidelities from the true Pauli fidelities. The solid lines show the $l_1$-distance of the (individually) learnable Pauli fidelities from the ground truth.
    }
    \label{fig:main_sim_intercept}
\end{figure}

\begin{figure}[t]
    \centering
    \subfloat[]{
    \centering
    \includegraphics[width=0.5\linewidth]{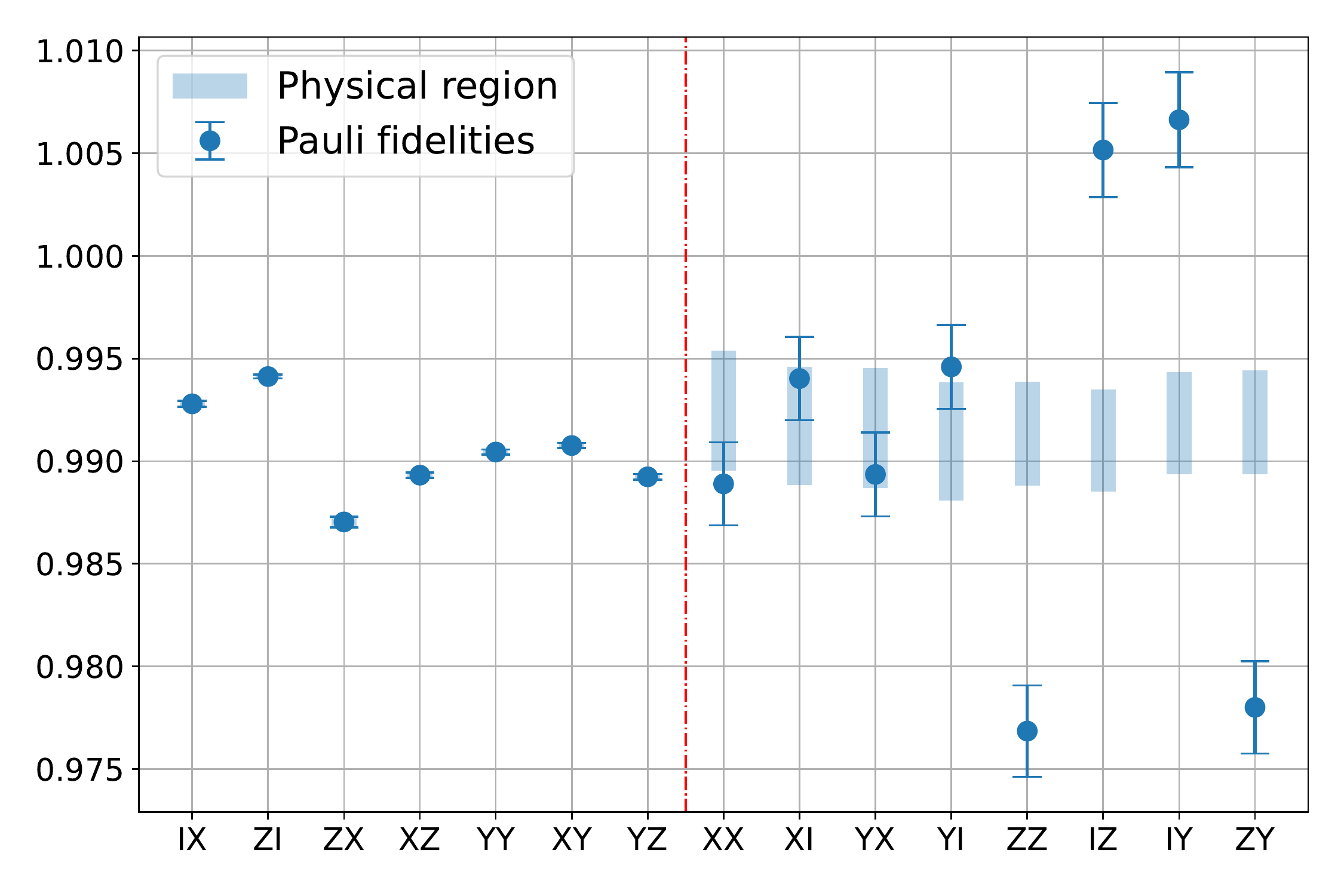}
    }
    \subfloat[]{
    \centering
    \includegraphics[width=0.5\linewidth]{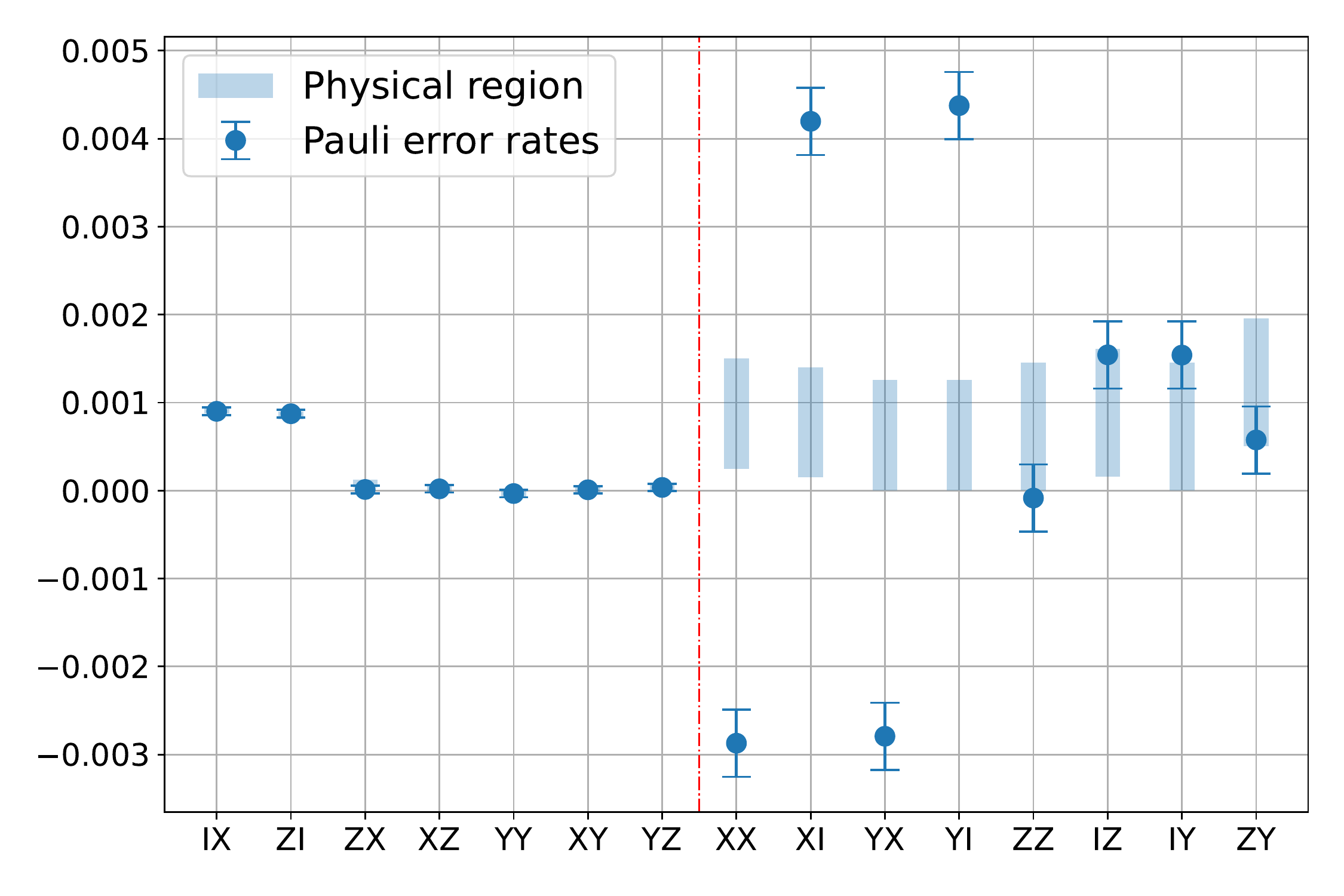}
    }
    \caption{The learned Pauli noise model using intercept CB. The feasible region (blue bars) are taken from Fig.~\ref{fig:main_exp_cbfeasible}. Estimates of Pauli fidelities (a) and Pauli error rates (b). Each data point is fitted using seven different circuit depths $L=[2,2^2,...,2^7]$. For each depth $C=150$ random circuits and $2000$ shots of measurements are used. Data are collected from \texttt{ibmq\_montreal} on 2022-03-23.}
    \label{fig:main_exp_intercept}
\end{figure}

\section{Discussion}
We have shown how to characterize the learnability of Pauli noise of Clifford gates and discussed a method to extract unlearnable information by assuming perfect initial state preparation. It is also interesting to consider other physically motivated assumptions on the noise model to avoid unlearnability. For example, we can write down a parameterization of the noise model based on the underlying physical mechanism which may have fewer than $4^n$ parameters. The main issue here is that these assumptions are highly platform-dependent and should be decided case-by-case. Moreover, it is unclear to what extent should the learned results be trusted when additional assumptions are made, since in general we cannot test whether the assumptions hold due to unlearnability.

Another direction to overcome the unlearnability is to change the model of quantum experiments. Here we have been working with the standard model as in gate set tomography, where a quantum measurement decoheres the system and only outputs classical information. However, some platforms might support quantum non-demolition (QND) measurements, and in this case measurements can be applied repeatedly, which could potentially allow more information to be learned~\cite{laflamme2022algorithmic}.

Recently, Ref.~\cite{huang2022foundations} considered similar issues of noise learnability. They studied a different Pauli noise model with perfect initial state $\ket{0}$, perfect computational basis measurement, and noisy single qubit gates, and showed the existence of unlearnable information. In contrast, here we focus on the learnability of Pauli noise of multi-qubit Clifford gates assuming perfect single-qubit gates (with noisy SPAM), and in practice we make the standard assumption that noise on single-qubit gates is gate-independent (\textit{e.g.}~\cite[Sec. II A]{Ferracin2022Efficiently}), in which case our noise learning results are interpreted as characterizing a dressed cycle.

This work leaves open the question of noise learnability for non-Clifford gates. An issue here is that randomized compiling is not known to work with non-Clifford gates in general, so it is unclear if the general CPTP noise learnability problem can be reduced to Pauli noise. Recent work~\cite{liu2021benchmarking} shows that random quantum circuits can effectively twirl the CPTP noise channel into Pauli noise and can be used to learn the total Pauli error. The question of whether more information can be learned still remains open.

Another issue to address is the scalability in noise learning. It is impossible to estimate all learnable degrees of freedom efficiently as there are exponentially many of them (an exponential lower bound on the sample complexity is shown in~\cite{chen2022quantum}). One way to avoid the exponential scaling issue is to assume the noise model has certain special structure (such as sparsity or low-weight) such that the noise model only has polynomially many parameters~\cite{harper2020efficient,harper2021fast,flammia2020efficient,berg2022probabilistic}. It is an interesting open direction to study the characterization of learnability under these assumptions, and we give some related discussions in Supplementary Section II D.

\section*{Data availability}
The data generated in this study is available at \url{https://github.com/csenrui/Pauli_Learnability}  

\section*{Code availability}
The code that supports the findings of this study is available at \url{https://github.com/csenrui/Pauli_Learnability} 

\begin{acknowledgments}
We thank Ewout van den Berg, Arnaud Carignan-Dugas, Robert Huang, Kristan Temme and Pei Zeng for helpful discussions.
We thank the anonymous reviewer \#2 for suggesting an alternative approach to intercept cycle benchmarking. 
S.C. and L.J. acknowledge support from the ARO (W911NF-18-1-0020, W911NF-18-1-0212), ARO MURI (W911NF-16-1-0349, W911NF-21-1-0325), AFOSR MURI (FA9550-19-1-0399, FA9550-21-1-0209), AFRL (FA8649-21-P-0781), DoE Q-NEXT, NSF (OMA-1936118, EEC-1941583, OMA-2137642), NTT Research, and the Packard Foundation (2020-71479).
Y.L. was supported by DOE NQISRC QSA grant \#FP00010905, Vannevar Bush faculty fellowship N00014-17-1-3025, MURI Grant FA9550-18-1-0161 and NSF award DMR-1747426. A.S. is supported by a Chicago Prize Postdoctoral Fellowship in Theoretical Quantum Science. B.F. acknowledges support from AFOSR (YIP number FA9550-18-1-0148 and FA9550-21-1-0008). This material is based upon work partially supported by the National Science Foundation under Grant CCF-2044923 (CAREER) and by the U.S. Department of Energy, Office of Science, National Quantum Information Science Research Centers. This research used resources of the Oak Ridge Leadership Computing Facility at the Oak Ridge National Laboratory, which is supported by the Office of Science of the U.S. Department of Energy under Contract No. DE-AC05-00OR22725.
\end{acknowledgments}

\section*{Author contributions}
S.C. and Y.L. developed the theory and performed the experiments. B.F. and L.J. supervised the project. All authors contributed important ideas during initial discussions and contributed to writing the manuscript.

\section*{Competing interests}
The authors declare no competing interests.

\bibliography{PauliNoise}

\newpage
\begin{appendix}

\date{\today}
\maketitle
\tableofcontents

\section{Preliminaries}\label{sec:pre}
Define ${\sf P}^n$ to be the $n$-qubit Pauli group modulo its center. 
We can label any Pauli operator $P_a\in{\sf P}^n$ with a $2n$-bit string $a$.
Specifically, we define $P_{\bm 0}$ to be the identity operator $I$.
We will use the notations $P_a$ and $a$ interchangeably when there is no confusion.

The \emph{pattern} of an $n$-qubit Pauli operator $P_a$, denoted as $\pt(P_a)$, is an $n$-bit string that takes $0$ at the $j$th bit if $P_a$ equals to $I$ at the $j$th qubit and takes $1$ otherwise. 
For example, $\pt(XYIZI)=\pt(XXIXI)=11010$.

\medskip

An $n$-qubit \emph{Pauli diagonal map} $\Lambda$ is a linear map of the following form
\begin{equation}
	\Lambda(\cdot) = \sum_{a\in{{\sf P}^n}}p_a P_a(\cdot)P_a,
\end{equation}
where $\bm{p}\coleq \{p_a\}_a$ are called the \emph{Pauli error rates}.
If $\Lambda$ is further a CPTP map, which corresponds to the condition $p_a\ge 0$ and $\sum_a p_a = 1$, then it is called a \emph{Pauli channel}.
An important property of Pauli diagonal maps is that their eigen-operators are exactly the $4^n$ Pauli operators. Thus, an alternative expression for $\Lambda$ is
\begin{equation}
	\Lambda(\cdot) = \frac{1}{2^n}\sum_{b\in{{\sf P}^n}}\lambda_b\Tr(P_b(\cdot))P_b,
\end{equation}
where $\bm\lambda \coleq \{\lambda_b\}_b$ are called the \emph{Pauli fidelities} or \emph{Pauli eigenvalues} ~\cite{flammia2020efficient,flammia2021pauli,chen2020robust}. 
These two sets of parameters, $\bm p$ and $\bm \lambda$, are related by the Walsh-Hadamard transform
\begin{equation}
	\begin{aligned}
		\lambda_b = \sum_{a\in{{\sf P}^n}}p_a(-1)^\expval{a,b},\quad
		p_a = \frac{1}{4^n}\sum_{b\in{{\sf P}^n}}\lambda_b(-1)^\expval{a,b},
	\end{aligned}
\end{equation}
where $\expval{a,b}$ equals to $0$ if $P_a,P_b$ commute and equals to $1$ otherwise.

\medskip

For a general linear map $\mc E$, define its \emph{Pauli twirl} as
\begin{equation}
    \mc E^{P}\coleq \sum_{a\in\Pn} \mc P_a\mc E \mc P_a.
\end{equation}
Here we use the calligraphic $\mc P_a$ to represent the unitary channel of Pauli gate $P_a$, $\mc P_a(\cdot) := P_a(\cdot)P_a$.
The Pauli twirl of any linear map (quantum channel) is a Pauli diagonal map (Pauli channel). When we talk about the Pauli fidelities of a non-Pauli channel, we are effectively referring to the Pauli fidelities of its Pauli twirl.

\section{Theory on the learnability of Pauli noise}

In this section, we give a precise characterization of what information in the Pauli noise channel associated with Clifford gates can be learned in the presence of state-preparation-and-measurement (SPAM) noise. 
Our results show that certain Pauli fidelities of a noisy multi-qubit Clifford gate cannot be learned in a SPAM-robust manner, even with the assumption that single-qubit gates can be perfectly implemented.
The proof is related to the notion of \emph{gauge freedom} in the literature of gate set tomography~\cite{nielsen2021gate}.
We note that the results presented in this section emphasizes on the no-go part, \textit{i.e.}, some information about the Pauli noise is (SPAM-robustly) unlearnable even with many favorable assumptions on the experimental conditions.
As shown in the main text, the learnable information about Pauli noise can be extracted in a much more practical setting using cycle benchmarking~\cite{erhard2019characterizing} and its variant.

\subsection{Assumptions and definitions}\label{sec:noisemodelAssumptions}

We focus on an $n$-qubit quantum system. Below are our assumptions on the noise model.

\begin{itemize}
    \item \textbf{Assumption 1.} All single qubit unitary operation can be perfectly implemented.

    \item \textbf{Assumption 2.} A set of multi-qubit Clifford gates $\mf G \coleq \{\mc G\}$ can be implemented and are subject to gate-dependent Pauli noise, \textit{i.e.}, $\widetilde{\mc G} = \mc G\circ \Lambda_{\mc G}$ where $\Lambda_\mc G$ is some $n$-qubit Pauli channel.

    \item \textbf{Assumption 3.} Any state preparation and measurement can be implemented, up to some fixed Pauli noise channel $\mc E^S$ and $\mc E^M$, respectively.

    \item\textbf{Assumption 4.} The Pauli noise channels appearing in the above assumptions satisfy that all Pauli fidelities and Pauli error rates are strictly positive.
    
\end{itemize}
    Assumption 1 is motivated by the fact that the noise of single-qubit gates are usually much smaller than that of multi-qubit gates on today's hardware. Such approximation is widely adopted in the literature~\cite{erhard2019characterizing,wallman2016noise} with slight modifications.
    In Assumption 2, we view every Clifford gate as an $n$-qubit gate, and allow the noise to be $n$-qubit. This means we are taking all crosstalk into account. 
    A Clifford gate acting on a different (ordered) subset of qubits is viewed as a different gate and can thus have a different noise channel (\textit{e.g.}, CNOT$_{12}$, CNOT$_{21}$, CNOT$_{23}$ have different noise channels.)
    We will discuss the no-crosstalk situation in Sec.~\ref{sec:no_crosstalk}.
    The rationale for assuming Pauli noise in Assumption 2 and 3 is that we can always use randomized compiling~\cite{wallman2016noise,hashim2020randomized} to tailor general noise into Pauli channels. 
    Finally, Assumption 4 is mostly for technical convenience. The requirement of positive Pauli error rates roughly implies the Pauli channels are at the interior of the CPTP polytope, and will be useful later in constructing valid gauge transformations. The requirement of positive Pauli fidelities is also reasonable for any physically interesting noise model.
    
    \medskip
    
    Specifying a Clifford gate set $\mf G$, a \emph{noise model} satisfying our assumptions is determined by the Pauli channels describing gate noise and SPAM noise. 
    We can thus view a noise model as a collection of Pauli fidelities, denoted as $\mc N = \{\mc E^S,\mc E^M,\Lambda\}$, where $\mc E^{S/M} = \{\lambda_a^{S/M}\}_a$ describes the SPAM noise and $\Lambda = \{\lambda_a^{\mc {G}}\}_{a,\mc G}$ describes the gate noise.
    We note that this is an example of \emph{parametrized gate set} in the language of gate set tomography~\cite{nielsen2021gate}.
    
    In order to gain information about an unknown noise model, one needs to conduct \emph{experiments}. In the circuit model, any experiment can be described by some state preparation, a sequence of quantum gates, and some POVM measurements. 
    An experiment conducted with different underlying noise model would yield different measurement outcome distributions.
    Explicitly, consider an (ideal) experiment with initial state $\rho_0$, gate sequence $\mc C$, POVM measurements $\{E_o\}_o$. Denote the noisy implementation of these objects within a certain noise model $\mc N$ with a tilde. Then the experiment effectively maps $\mc N$ to a probability distribution $p_{\mc N}(o) = \Tr(\widetilde E_o(\widetilde{\mc C}(\widetilde \rho_0)))$. 
    We call two noise models $\mc N_1$, $\mc N_2$ \emph{indistinguishable} if for all possible experiments we have $p_{\mc N_1}=p_{\mc N_2}$, and distinguishable otherwise.
    
    \begin{definition}[Learnable and unlearnable function]\label{de:learnability}
      A function $f$ of noise models is learnable if
      \begin{equation}
        f(\mc N_1)\ne f(\mc N_2) \implies \mc N_1, \mc N_2~\text{are distinguishable},
      \end{equation} 
      for any noise models $\mc N_1$, $\mc N_2$. 
      In contrast, $f$ is unlearnable if there exist indistinguishable noise models $\mc N_1$, $\mc N_2$ such that $f(\mc N_1)\ne f(\mc N_2)$. 
    \end{definition}
    
    Note that the above definition of ``learnable'' does not necessarily mean that the value of the function can be learned. However, throughout this paper whenever some function is ``learnable'' according to Definition~\ref{de:learnability}, it is also learnable in the stronger sense that we can design an experiment to estimate it up to arbitrarily small error with high success probability.
    
    In the language of gate set tomography, an unlearnable function is a \emph{gauge-dependent} quantity of the gate set~\cite{nielsen2021gate}.
    On the other hand, any learnable function can in principle be learned to arbitrary precision.
    In the following, we will focus the learnability of the functions of the gate noise, including individual and multiplicative combinations of Pauli fidelities.
    
\medskip

\subsection{Learnability of individual Pauli fidelity}

We first study the learnability of individual Pauli fidelities associated with a Clifford gate. This has been an open problem in recent study of quantum benchmarking.
Perhaps surprisingly, we obtain the following simple criteria on the learnability of Pauli fidelities with any Clifford gate.

\begin{theorem}\label{th:nogo}
With Assumptions 1-4, for any $n$-qubit Clifford gate $\mc G$ and Pauli operator $P_a$, the Pauli fidelity $\lambda_a^{\mc G}$ is unlearnable if and only if $\mc G$ changes the pattern of $P_a$, \textit{i.e.,} $\pt(\mc G(P_a))\ne \pt(P_a)$.
\end{theorem}

The fact that certain Pauli fidelities are SPAM-robustly unlearnable is observed in some recent works~\cite{erhard2019characterizing,hashim2020randomized,berg2022probabilistic,Ferracin2022Efficiently}, described as ``degeneracy'' of the noise model. Our work is the first to give a rigorous argument for this by establishing connections to gate set tomography.
As an example, for the CNOT and SWAP gates, we can immediately list its learnable and unlearnable Pauli fidelities in Table~\ref{tab:cnot_swap_individual}.
We note that, the no-go theorem holds even under the no-crosstalk assumption as will be discussed in Sec.~\ref{sec:no_crosstalk},
so introducing ancillary qubits or other multi-qubit Clifford gates cannot help resolve the unlearnability.

\begin{table}[!htp]
    \centering
    \begin{tabular}{|c|c|c|}
        \hline
        Gate & Learnable & Unlearnable \\
        \hline
        CNOT&$\lambda_{II},\lambda_{ZI},\lambda_{IX},\lambda_{ZX},
        \lambda_{XZ},\lambda_{YY},\lambda_{XY},\lambda_{YZ}$ &
        $\lambda_{IZ},\lambda_{XI},\lambda_{ZZ},\lambda_{XX},
        \lambda_{IY},\lambda_{YI},\lambda_{ZY},\lambda_{YX}$\\
        \hline
        SWAP&$\lambda_{II},\lambda_{XX},\lambda_{XY},\lambda_{XZ},
        \lambda_{YX},\lambda_{YY},\lambda_{YZ},\lambda_{ZX},\lambda_{ZY},\lambda_{ZZ}$ &
        $\lambda_{IX},\lambda_{IY},\lambda_{IZ},
        \lambda_{XI},\lambda_{YI},\lambda_{ZI}$\\
        \hline
    \end{tabular}
    \caption{Learnability of individual Pauli fidelity of CNOT and SWAP.}
    \label{tab:cnot_swap_individual}
\end{table}

Before going into the proof, we make several remarks about Theorem~\ref{th:nogo}.
The correct interpretation of the no-go result in Theorem~\ref{th:nogo} is that certain Pauli fidelities cannot be learned in a fully SPAM-robust manner.
If one has some pre-knowledge that the SPAM noise is much weaker than the gate noise, there exist methods to give a pretty good estimate of those unlearnable Pauli fidelities, according to physical constraints.
See the discussions in the main text.
On the other hand, it is observed that the product of certain unlearnable Pauli fidelities can be learned in a SPAM-robust manner, such as $\lambda_{XI}\cdot\lambda_{XX}$ for the CNOT gate~\cite{erhard2019characterizing}. We will characterize the learnability of this kind of products of Pauli fidelities in the next subsection.

\begin{proof}[Proof of Theorem~\ref{th:nogo}]
    We start with the ``only if'' part, which is equivalent to saying that $\pt(P_a)=\pt(\mc G(P_a))$ implies $\lambda_a^{\mc G}$ being learnable.
    The condition $\pt(\mc G(P_a))=\pt(P_a)$ implies $\mc G(P_a)$ is equivalent to $P_a$ up to some local unitary transformation, \textit{i.e.}, there exists a product of single-qubit unitary gates $\mc U \coleq \bigotimes_{j=1}^n\mc U_j$ such that
    \begin{equation}
        \mc U\circ\mc G(P_a) = P_a.
    \end{equation}
    Now we design the following experiments parameterized by a positive integer $m$,
\begin{itemize}
    \item Initial state: $\rho_0 = (I+P_a)/2^n$,
    \item POVM measurement: $E_{\pm 1} = (I\pm P_a)/{2}$,
    \item Circuit: $\mc C^m = \left(\mc U\circ\mc G\right)^m$.
\end{itemize}
Consider the measurement probability by running these experiments within a noise model $\mc N$.
\begin{equation}
    \begin{aligned}
      p^{(m)}_{\pm 1}(\mc N) &= \Tr\left( \widetilde{E}_{\pm 1} \widetilde{\mc C}^m  (\widetilde{\rho}_0)\right) \\
      &= \Tr \left( \frac{I\pm P_a}{2} \cdot\left( \mc E^M \circ \left( 
        \mc U\circ\mc G      
      \right)^m\circ\mc E^S \right)
      \left(\frac{I+P_a}{2^n}\right) \right)\\
      &= \Tr\left(\frac{I\pm P_a}{2} \cdot\frac{I+\lambda^M_{a}
      \left(\lambda_a^{\mc G}\right)^m
      \lambda^S_{a}P_a}{2^n} \right)\\
      &= \frac{1\pm\lambda^M_{a}
      \left(\lambda_a^{\mc G}\right)^m
      \lambda^S_{a}}{2}.
    \end{aligned}
\end{equation}
Recall that $\lambda_a^{S/M}$ is the Pauli fidelity of the SPAM noise channel for $P_a$.
The expectation value is
\begin{equation}
    \mbb E^{(m)}(\mc N) = \lambda^M_{a}
     \left(\lambda_a^{\mc G}\right)^m
      \lambda^S_{a}.
\end{equation}
If we take the ratio of expectation values of two experiments with consecutive $m$, we obtain (recall that all these Pauli fidelities are strictly positive by Assumption 4)
\begin{equation}
    \mbb E^{m+1}(\mc N)/\mbb E^{m}(\mc N) = \lambda_a^{\mc G}.
\end{equation}
This implies that if two noise model assign different values for $\lambda_a^{\mc G}$, the above experiments would be able to distinguish between them. By definition~\ref{de:learnability}, we conclude $\lambda_a^{\mc G}$ is learnable.

    \bigskip
    
    Next we prove the ``if'' part. Fix an $n$-qubit Clifford gate $\mc G$.
    Let $P_a$ be any Pauli operator such that $\pt(\mc G(P_a))\neq \pt(P_a)$. We will show that $\lambda_a^{\mc G}$ is unlearnable by explicitly constructing indistinguishable noise models that assign different values to $\lambda_a^{\mc G}$.

    Recall that any experiment involves a noisy initial state $\tilde{\rho}_0$, a noisy measurement $\{\widetilde{E}_l\}_l$, and a quantum circuit consisting of noiseless single-qubit gates $\mc U \coleq \bigotimes_{j=1}^n\mc U_j$ and noisy multi-qubit Clifford gates $\widetilde{\mc T}$.
    Now, introduce an invertible linear map $\mc M:\mc L(\mc H_{2^n})\to \mc L(\mc H_{2^n})$, and consider the following transformation 
    \begin{equation}\label{eq:gauge_trans}
        \begin{aligned} &\widetilde{\rho}_0\mapsto \mc M(\widetilde{\rho}_0),\quad
            \widetilde{E}_l\mapsto (\mc M^{-1})^\dagger (\widetilde{E}_l),\\
            &\bigotimes_{j=1}^n\mc U_j \mapsto \mc M\circ\bigotimes_{j=1}^n\mc U_j\circ\mc M^{-1},\\
            &\widetilde{\mc T} \mapsto \mc M\circ \widetilde{\mc T}\circ\mc M^{-1}.
        \end{aligned}
    \end{equation}
    One can immediately see that any measurement outcome distribution $p_l\coleq\Tr(\widetilde{E}_l\widetilde{\mc C}(\widetilde{\rho}_0))$ remains unchanged via such transformation. Therefore the noise models related by this transformation are indistinguishable. This is called a \emph{gauge transformation} in the literature of gate set tomography~\cite{nielsen2021gate}. 
    To use this idea for the proof, we start with a noise model $\mc N$ and construct a map $\mc M$ such that
    \begin{enumerate}
        \item The transformation yields a physical noise model $\mc N'$ satisfying Assumptions 1-5 in Sec.~\ref{sec:noisemodelAssumptions}.
        \item The two noise models $\mc N$, $\mc N'$ assign different values to $\lambda_a^{\mc G}$.
    \end{enumerate}

    Starting with a generic noise model $\mc N = \{\mc E^S,\mc E^M,\Lambda\}$ satisfying the assumptions,
    we construct the gauge transform map $\mc M$ as follows.
    Since $\pt(\mc G(P_a))\ne \pt(P_a)$, there exists an index $i\in[k]$ such that one and only one of $(P_a)_i$ and $\mc G(P_a)_i$ equals to $I$. Let $\mc M$ be the single-qubit depolarizing channel on the $i$-th qubit, 
    \begin{equation}\label{eq:dep_trans}
            \mc M\coleq\mc D_{i}\otimes\mc I_{[n]\backslash i},
    \end{equation}
    where the single-qubit depolarizing channel is defined as
    \begin{equation}
        \forall P\in\{I,X,Y,Z\},\quad \mc D(P) = \left\{
        \begin{aligned}
        P,\quad&~\text{if}~ P=I,\\
        \eta P,\quad&~\text{otherwise},
        \end{aligned}
        \right.
    \end{equation}
    for some parameter $0<\eta<1$. We will specify the value of $\eta$ later.
    
    Now we calculate the transformed noise model $\mc N'=\{\mc E^{S'},\mc E^{M'},\Lambda'\}$. The SPAM noise channels are transformed as
    \begin{equation}\label{eq:SPAMupdate}
        \mc E^{S'} = \mc M\mc E^S,\quad \mc E^{M'} = \mc E^M\mc M^{-1},
    \end{equation}
    both of which are still Pauli diagnoal maps. Thanks to our Assumption 4, as long as $\eta$ is sufficiently close to $1$, they can be shown to be Pauli channels.

    Next, the single-qubit unitary gates are transformed as
    \begin{equation}
        \mc M \left(\bigotimes_{j=1}^n \mc U_j \right)\mc M^{-1}
        =\mc D_i\mc U_i\mc D_i^\dagger \otimes \bigotimes_{j\ne i}\mc U_j
        = \bigotimes_j \mc U_j,
    \end{equation}
    since the single-qubit deplorizing channel commutes with any single-qubit unitary. This implies the single-qubit unitary gates are still noiseless.

    Finally, consider an arbitrary $n$-qubit Clifford gate ${\mc T}$. We show that the transformed noisy gate takes the form $\widetilde{\mc T}'=\widetilde{\mc T}\circ \Lambda_{\mc T}'$ where $\Lambda_{\mc T}'$ is still a Pauli channel, with the Pauli fidelities updated as follows.
    \begin{equation}\label{eq:paulifidelityupdate}
        {\lambda_b^{\mc T}}'=\begin{cases}
            \eta \lambda_b^{\mc T}, & \text{if }\pt(P_b)_i=0\text{ and }\pt(\mc T(P_b))_i=1,\\
            \eta^{-1}\lambda_b^{\mc T}, & \text{if }\pt(P_b)_i=1\text{ and }\pt(\mc T(P_b))_i=0,\\
            \lambda_b^{\mc T}, &\text{if }\pt(P_b)_i=\pt(\mc T(P_b))_i.
        \end{cases}
    \end{equation}
    We give a proof for the first case. Note that
    \begin{equation}\label{eq:gate_noise_trans}
        \begin{aligned}
        \mc M\circ \widetilde{\mc T}\circ\mc M^{-1}
        &= \mc D_{i}\circ \widetilde{\mc T}\circ \mc D_{i}^{-1}\\
            &= \mc D_{i}\circ {\mc T}\circ\Lambda_{\mc T}\circ \mc D_{i}^{-1}\\
            &= {\mc T}\circ({\mc T}^{-1}\circ\mc D_{i}\circ {\mc T}\circ\Lambda_{\mc T}\circ \mc D_{i}^{-1}) \\
            &\eqcol\mc T\circ \Lambda'_{\mc T},
        \end{aligned}
    \end{equation}
    where we use $\mc D_i$ as a shorthand for $\mc D_i\otimes\mc I_{[n]\backslash i}$.
    The transformed noise channel can be written as
    \begin{equation}\label{eq:noisechannelupdate}
        \Lambda'_{\mc T}={\mc T}^{-1}\circ\mc D_{i}\circ {\mc T}\circ\Lambda_{\mc T}\circ \mc D_{i}^{-1}.
    \end{equation}
    Let us calculate its action on arbitrary $P_b$.
    \begin{equation}
        \begin{aligned}
            \Lambda_{\mc T}'
            ({P_{b}}) &= (\mc T^{-1}\circ \mc D_{i}\circ \mc T\circ\Lambda_{\mc T}\circ \mc D_{i}^{-1})(P_b)\\
            &= (\eta^{-1})^{\pt(P_b)_i} (\mc T^{-1}\circ \mc D_{i}\circ \mc T\circ\Lambda_{\mc T})(P_b)\\
            &= \lambda_b^{\mc T}(\eta^{-1})^{\pt(P_b)_i} (\mc T^{-1}\circ \mc D_{i}\circ \mc T)(P_b)\\
            &=\eta^{\pt(\mc T(P_b))_i}\lambda_b^{\mc T}(\eta^{-1})^{\pt(P_b)_i} ~P_b.
        \end{aligned}
    \end{equation}
    Thus, $\Lambda_{\mc T}'$ is indeed a Pauli diagonal map with Pauli fidelities given by Eq.~\eqref{eq:paulifidelityupdate}. The fact that $\Lambda_{\mc T}'$ is guaranteed to be a CPTP map by choosing appropriate $\eta$ will be verified later.
    Specifically, if we take $\mc T$ to be the Clifford gate $\mc G$ that we are interested in, we have $\lambda_a^{\mc G'} = \eta\lambda_a^{\mc G}$ or $\lambda_a^{\mc G'} = \eta^{-1}\lambda_a^{\mc G}$. In either case, $\lambda_a^{\mc G'} \ne \lambda_a^{\mc G}$. This means the two indistinguishable noise model $\mc N$, $\mc N'$ indeed assign different values to $\lambda_a^{\mc G}$.
    
    We now verify that $\mc N'$ is indeed a physical noise model and satisfies Assumptions 1-4. We have already shown that single-qubit unitary gates remain noiseless and that all gate noise and SPAM noise are described by Pauli diagonal maps. The only thing left is to make sure all these Pauli diagonal maps are CPTP and satisfy the positivity constraints in Assumption 4. 
    According to Eq.~\eqref{eq:SPAMupdate} and \eqref{eq:paulifidelityupdate}, any Pauli fidelity $\lambda_b$ of either SPAM noise or gate noise is transformed to one of the following $\lambda_b'\in\{\lambda_b,\eta\lambda_b,\eta^{-1}\lambda_b\}$, so $\lambda_b>0$ implies $\lambda_b'>0$. On the other hand, any transformed Pauli error rate can be bounded by
    \begin{equation}\label{eq:positivity}
        \begin{aligned}
            p_c' &= \frac{1}{4^n}\sum_{b\in{\sf P}^n}(-1)^\expval{b,c}\lambda_b'\\
            &\ge \frac{1}{4^n}\sum_{b\in{\sf P}^n}\left((-1)^\expval{b,c}\lambda_b - (\eta^{-1}-1)\lambda_b\right)\\
            &\ge p_c - (\eta^{-1}-1).
        \end{aligned}
    \end{equation}
    To ensure every $p'_c>0$, we can choose $1>\eta>(p_{\min}+1)^{-1}$ with $p_{\min}$ being the minimum Pauli error rate among all Pauli channels of both SPAM and gate noise, which is possible since $p_{\min}>0$ by Assumption 4. This means each transformed Pauli diagonal maps are completely positive (CP). To see they are also trace-preserving (TP), just notice from Eq.~\eqref{eq:SPAMupdate},~\eqref{eq:paulifidelityupdate} that $\lambda_{\bm 0}'=\lambda_{\bm 0} = 1$ always holds. 
    Now we conclude that $\mc N'$ is indeed a physical noise model satisfying all the assumptions. 
    Combining with the reasoning in the last paragraph, we see  $\lambda_a^{\mc G}$ is unlearnable.
    This completes our proof.
    \end{proof}
    
\subsection{Characterization of learnable space via algebraic graph theory}\label{sec:space}

We have characterized the learnability of individual Pauli fidelities associated with any Clifford gates in Theorem~\ref{th:nogo}.
Here, we want to understand the learnablity for a general function of the gate noise.
We first show that, in our setting, any measurement outcome probability in experiment can be expressed as a polynomial of Pauli fidelities of gate and SPAM noise, and each term in the polynomial can be learned via a CB experiment (see Sec.~\ref{sec:justification} for details).
Therefore, it suffices to study the monomials, \textit{i.e.}, products of Pauli fidelities.
For each Pauli fidelity $\lambda_a^{\mc G}$, we define the \emph{logarithmic Pauli fidelity} as $l_a^{\mc G}\coleq \log \lambda_a^{\mc G}$ ($\lambda_a^{\mc G}>0$ by Assumption 4). 
It then suffices to study the learnability of linear functions of the logarithmic Pauli fidelities. 
An alternative reason to only study this class of function is that, under a weak noise assumption, we have $l_a\to 0$, so we can express any function of the noise model as a linear function of $l_a$ under a first order approximation. Note that similar approaches have been explored in the literature~\cite{nielsen2022first,flammia2021averaged}.

Since we are working with Assumption 1-4 which takes all crosstalk into account, we treat the noise channel for each gate in $\mf G$ as $n$-qubit. 
The number of independent Pauli fidelities we are interested in is thus 
\begin{equation}
    |\Lambda| = |\mf G|\cdot 4^n.    
\end{equation}
Denote the space of all (real-valued) linear function of logarithmic Pauli fidelities as $F$, then we have $F\cong \mbb R^{|\Lambda|}$. A function $f\in F$ uniquely corresponds to a vector $\bm v\in \mbb R^{|\Lambda|}$ by $f(\bm l) = \bm v\cdot \bm l = \sum_{a,\mc G}v_{a,\mc G}l_a^{\mc G}$. We will use the vector to refer to the linear function when there is no ambiguity.

Denote the set of all learnable function in $F$ as $F_L$ (in the sense of Def.~\ref{de:learnability}). As shown in the following lemma, $F_L$ forms a linear subspace in $F$, so we call $F_L$ the \emph{learnable space}.

\begin{lemma}\label{le:learnable_is_space}
  $F_L$ is a linear subspace of $F$.
\end{lemma}
\begin{proof}
  Given $\bm v_1,\bm v_2 \in F_L$, consider the learnability of $\bm v_1+\bm v_2$. For any noise models $\mc N_1,~\mc N_2$, 
  \begin{equation}
  \begin{aligned}
    (\bm v_1+\bm v_2)\cdot\bm l_{\mc N_1} \ne (\bm v_1+\bm v_2)\cdot\bm l_{\mc N_2} &\implies \bm v_1\cdot\bm l_{\mc N_1} \ne  \bm v_1\cdot\bm l_{\mc N_2} ~\text{or}~\bm v_2\cdot\bm l_{\mc N_1} \ne \bm v_2\cdot\bm l_{\mc N_2}\\
    &\implies \mc N_1, \mc N_2~\text{are distinguishable}.
  \end{aligned}
  \end{equation}  
  Thus $\bm v_1+\bm v_2\in F_L$.
  We also have $\bm v\in L \implies k\bm v\in F_L$ for all $k\in\mbb R$. Therefore, $F_L$ forms a vector space in $\mbb R^{|\Lambda|}$.
\end{proof}

Our goal is to give a precise characterization of the learnable space $F_L$. 
For example, we may want to know the dimension of $F_L$, which represents the learnable degrees of freedom for the noise.
This is also the maximum number of linearly-independent equations about the logarithmic Pauli fidelities we can expect to extract from experiments.
Conversely, the unlearnable degrees of freedom roughly correspond to the number of independent gauge transformations.
We summarize these definitions as follows.

\begin{definition}
    Given a Clifford gate set $\mf G$, the learnable degrees of freedom $\mr{LDF}(\mf G)$ and unlearnable degrees of freedom $\mr{UDF}(\mf G)$ are defined as, respectively,
    \begin{equation}
        \mr{LDF}(\mf G) \coleq \mr{dim}(F_L),\quad
        \mr{UDF}(\mf G) \coleq |\Lambda| - \mr{dim}(F_L).
    \end{equation}
\end{definition}

\medskip
\noindent Our approach is to relate $F_L$ to certain properties of a graph defined as follows.

\begin{definition}[Pattern transfer graph]
  The pattern transfer graph associated with a Clifford gate set $\mf G$ is a directed graph $G=(V,E)$ constructed as follows:
\begin{itemize}
    \item $V(G) = \{0,1\}^n$.
    \item $E(G) = \{e_{a,\mc G} \coleq (\pt(P_a),~\pt(\mc G(P_a)) ~|~\forall~ P_a\in{\sf P}^n,~\mc G\in\mf{G} \}$.
\end{itemize}
\end{definition}
The $2^n$ vertices each corresponds to a possible Pauli pattern.
The $|E| = |\Lambda| = |\mf G|\cdot 4^n$ edges each corresponds to a Pauli operator and a Clifford gate, describing how the Clifford gate evolves the pattern of the Pauli operator. One can also think each edge corresponds to a unique Pauli fidelity ($e_{a,\mc G}\leftrightarrow \lambda_{a}^{\mc G}$).
The rationale for only tracking the Pauli pattern is that we assume the ability to implement noiseless single-qubit unitaries, which makes the actual single-qubit Pauli operators unimportant. Fig.~2 of main text shows the pattern transfer graphs for a CNOT gate, a SWAP gate, and a gate set of CNOT and SWAP, respectively.

\medskip

Next, we give some definitions from graph theory (see~\cite{gleiss2003circuit,bollobas1998modern}). A \emph{chain} is an alternating sequences of vertices and edges $z=(v_0,e_1,v_1,e_2,v_2,...,v_{q-1},e_q,v_q)$ such that each edge satisfies $e_k=(v_{k-1},v_k)$ or $e_k=(v_{k},v_{k-1})$. 
A chain is \emph{simple} if it does not contain the same edge twice.
A closed chain (\textit{i.e.}, $v_0=v_q$) is called a \emph{cycle}. If an edge $e_k$ in a chain satisfies $e_k = (v_{k-1},v_k)$, it is called an \emph{oriented edge}. A chain consists solely of oriented edges is called a \emph{path}. A closed path is called a \emph{oriented cycle} or a \emph{circuit}.
A graph is called \emph{strongly connected} if there is a path from every vertex to every other vertex. A graph is called \emph{weakly connected} if there is a chain from every vertex to every other vertex. The number of (strongly or weakly) \emph{connected components} is the minimum number of partitions of the vertex set $V=V_1\cup\cdots\cup V_c$ such that each subgraph generated by a vertex partition is (strongly or weakly) connected.

We can equip a graph with vector spaces.
Following the notations of~\cite[Sec. II.3]{bollobas1998modern}, 
the \emph{edge space} $C_1(G)$ of a directed graph $G$ is the vector space of all linear functions from the edges $E(G)$ to $\mbb R$.
By construction, $C_1(G)\cong \mbb R^{|\Lambda|} \cong F$. 
Every linear function of the logarithmic Pauli fidelities naturally corresponds to a linear function of the edges according to the label of the edges ($l_{a}^{\mc G} \leftrightarrow e_{a,\mc G}$).
Again, we use vectors in $\mbb R^{|\Lambda|}$ to refer to elements of $C_1(G)$. The inner product on $C_1(G)$ is defined as the standard inner product on $\mbb R^{|\Lambda|}$.

There are two subspaces of $C_1(G)$ that is of special interest. 
For a simple cycle 
$z$ in $G$, we assign a vector $\bm v_z\in C_1(G)$ as follows
\begin{equation}
    \bm v_z(e) = \left\{
    \begin{aligned}
      +1,\quad& e\in z,~\text{$e$ is oriented.}\\
      -1,\quad& e\in z,~\text{$e$ is not oriented.}\\
      0,\quad& e\notin z.
    \end{aligned}
    \right.
\end{equation}
The \emph{cycle space} $Z(G)$ is the linear subspace of $C_1(G)$ spanned by all cycles $\bm v_z$ in $G$. 

Given a partition of vertices $V=V_1\cup V_2$ such that there is at least one edge between $V_1$ and $V_2$, a \emph{cut} is the set of all edges $e = (u,v)$ such that one of $u,v$ belongs to $V_1$ and the other belongs to $V_2$. For each cut $p$ we assign an vector $\bm v_p\in C_1(G)$ as follows
\begin{equation}\label{eq:cut}
    \bm v_p(e) = \left\{
    \begin{aligned}
      +1,\quad& e\in p,~\text{$e$ goes from $V_1$ to $V_2$.}\\
      -1,\quad& e\in p,~\text{$e$ goes from $V_2$ to $V_1$.}\\
      0,\quad& e\notin p.
    \end{aligned}
    \right.
\end{equation}
The \emph{cut space} $U(G)$ is the linear subspace of $C_1(G)$ spanned by all cuts $\bm v_p$ in $G$. 
Note that different partition of vertices may result in the same cut vector if $G$ is unconnected.

\begin{lemma}\cite[Sec. II.3, Theorem~1]{bollobas1998modern}\label{le:complement}
The edge space $C_1(G)$ is the orthogonal direct sum of the cycle space $Z(G)$ and the cut space $U(G)$, whose dimensions are given by
\begin{equation}
    \mr{dim}(Z(G)) = |E|-|V|+c(G),\quad 
    \mr{dim}(U(G)) = |V|-c(G),
\end{equation}
where $c(G)$ is the number of weakly connected components of $G$.
\end{lemma}

In some cases, we are more interested in circuits (oriented cycles) instead of general cycles. The following lemma gives a sufficient condition when the cycle spaces have a circuit basis, \textit{i.e.} the cycle space is spanned by oriented cycles.

\begin{lemma}\cite[Theorem~7]{gleiss2003circuit}\label{le:circuit}
    A directed graph has a circuit basis if it is strongly connected, or it is a union of strongly connected subgraphs. 
\end{lemma}

\medskip

With all the graph theoretical tools introduced above, we are ready to present the main result of this section.

\begin{theorem}\label{th:space}
  Under the Assumptions 1-4.
  For any $\mf G$,
  $F_L \cong Z(G)$.
  Explicitly, a linear function $f_{\bm v}(\bm l) = \bm v\cdot\bm l$ is learnable if and only if $\bm v$ belongs to the cycle space $Z(G)$.
\end{theorem}

We give the proof at the end of this section.
The proof involves two parts. 
The first is to show that every cycle is learnable using a variant of cycle benchmarking~\cite{erhard2019characterizing}, thus the cycle space belongs to the learnable space. The second part is to show that every cut induces a gauge transformation~\cite{nielsen2021gate}, and thus the learnable space must be orthogonal to the cut space, which implies it lies in the cycle space.

We remark that Theorem~\ref{th:nogo} can be viewed as a corollary of Theorem~\ref{th:space}. 
This is because an individual Pauli fidelity $\lambda_a^{\mc G}$ whose Pauli pattern changes (\textit{i.e.}, $\pt(P_a)\ne\pt(\mc G(P_a))$) corresponds to an simple edge in the pattern transfer graph, which does not belong to the cycle space and is thus unlearnable. On the other hand, a Pauli fidelity without Pauli pattern change corresponds to a self-loop in the pattern transfer graph, which belongs to the cycle space by definition, and is thus learnable.

\medskip

Combing Theorem~\ref{th:space} with Lemma~\ref{le:complement} leads to the following.
\begin{corollary}\label{co:udf}
    The learnable and unlearnable degrees of freedom associated with $\mf G$ are given by
    \begin{equation}
        \mr{LDF}(\mf G) = |\mf G|\cdot 4^n-2^n+c(\mf G),\quad
        \mr{UDF}(\mf G) = 2^n-c(\mf G),
    \end{equation}
    where $c(\mf G)$ is the number of connected components of the pattern transfer graph associated with $\mf G$.
\end{corollary}

Note that the unlearnable degrees of freedom always constitute an exponentially small portion, though they can grow exponentially. 

Examples of some gate sets are given in Table~\ref{tab:UDF} and Figure~\ref{fig:more_gate_sets}. One can notice some interesting properties. The UDF of CNOT and SWAP equals to $2$ and $1$, respectively, but a gate set containing both has $\mr{UDF}=2$. This means UDF is not ``additive''. The interdependence between different gates can give us more learnable degrees of freedom.
However, Corollary~\ref{co:udf} implies that the UDF of a gate set cannot be smaller than the UDF of any of its subset. 
This is because adding new gates can only decrease the number of connected components $c(\mf G)$ of the pattern transfer graph.

\begin{table}[h]
    \centering
    \begin{tabular}{|c|c|c|c|}
        \hline
        Number of qubits $n$ & Gate set $\mf G$ & Number of parameters $|\Lambda|=4^n|\mf G|$ &  $\mr{UDF}(\mf G)$\\
        \hline
        2 & CNOT & 16 & 2\\
        2 & SWAP & 16 & 1\\        
        2 & \{CNOT, SWAP\} & 32 & 2\\
        3 & $\mr{\{CNOT_{12},CNOT_{23},CNOT_{31}\}}$ & 192 & 6\\
        3 & $\mr{CIRC_3}$ & 64 & 4\\
        \hline
    \end{tabular}
    \caption{The unlearnable degrees of freedom of some gate sets. Here $\mr{CIRC_3}$ is the circular permutation on $3$ qubits. UDF is calculated by applying Corollary~\ref{co:udf} to the corresponding pattern transfer graph in Fig.~2 of main text and Fig.~\ref{fig:more_gate_sets}.}
    \label{tab:UDF}
\end{table}

\begin{figure}[htp]
    \centering
    \includegraphics[width=\columnwidth]{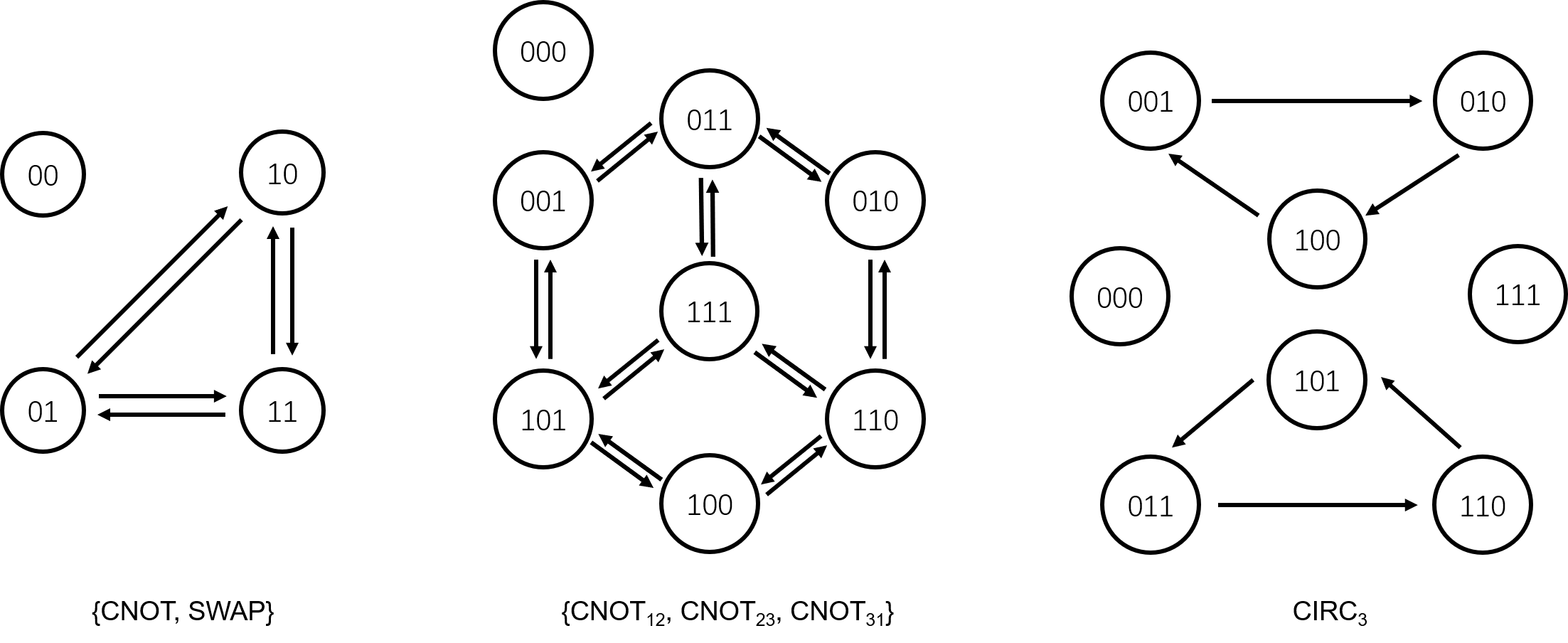}
    \caption{Pattern transfer graphs for $\mr{\{CNOT,~SWAP\}}$, $\mr{\{CNOT_{12},CNOT_{23},CNOT_{31}\}}$, and $\mr{CIRC}_3$. 
    For clarity, we omit labels of the edges, multiple edges, and self-loop. These omissions do not change the cut space of the graph. }
    \label{fig:more_gate_sets}
\end{figure}

\begin{proof}[Proof of Theorem~\ref{th:space}]
The proof is divided into showing $Z(G)\subseteq F_L$ and $F_L\subseteq Z(G)$ (up to the natural isometry between $F$ and $C_1(G)$).

\medskip

$Z(G)\subseteq F_L$: Roughly, this is equivalent to saying that all cycles are learnable. 
We will first show that the pattern transfer graph always has a circuit basis, and then show that the linear function associated with each circuit can be learned using a variant of cycle benchmarking protocol~\cite{erhard2019characterizing}.

We begin by showing that the pattern transfer graph $G$ associated with a gate set $\mf G$ is a union of strongly connected subgraphs. This is equivalent to saying that for any vertices $u,v\in V(G)$, if there is a path from $u$ to $v$, there must be a path from $v$ to $u$.
It suffices to show that for each edge $e=(u,v)$ there is a path from $v$ to $u$, since any path is just concatenation of edges.
By definition, the existence of $e=(u,v)$ implies there exists $P\in {\sf P}^n$ and $\mc G\in\mf G$ such that $\pt(P) = u$ and $\pt(Q) = v$ where $Q\coleq \mc G(P)$. Since a Clifford gate is a permutation on the Pauli group, there must exist some integer $d>0$ such that $\mc G^d = \mc I$, thus $P = \mc G^{d-1}(Q)$, which induces the following path from $v$ to $u$:
\begin{equation*}
(
    \pt(Q),~e_{Q,\mc G},~\pt(\mc G(Q)),~e_{\mc G(Q),\mc G},~ \pt(\mc G^2(Q)),~\cdots,~ \pt(\mc G^{d-2}(Q)),~ e_{\mc G^{d-1}(Q),\mc G},~ \pt(\mc G^{d-1}(Q))
).
\end{equation*}
One can verify this is a path according to the definition of $G$. This shows that $G$ is indeed a union of strongly connected subgraphs. According to Lemma~\ref{le:circuit}, $G$ has a circuit basis that spans the cycle space $Z(G)$.

Now we show that every circuit in $G$ represents a learnable function. 
Consider an arbitrary circuit $z = (v_0,e_1,v_1,e_2,v_2,...,v_{q-1},e_q,v_q\equiv v_0)$. For each $k=1...q$, the edge $e_k$ corresponds to a Pauli operator $P_k\in{\sf P}^n$ and a Clifford gate $\mc G_k \in \mf G$ such that $\pt(P_k) = v_{k-1}$ and $\pt(Q_k) = v_{k}$ where $Q_k\coleq \mc G_k(P_k)$.
On the other hand, since $\pt(Q_k)=\pt(P_{k+1})$, there exists a product of single qubit unitaries $\mc U_k$ such that $P_{k+1} = \mc U_k(Q_k)$ for $k=1...q$ (where we define $P_{q+1}\coleq P_1$, as $\pt(Q_q)=\pt(P_1)$ by assumptions). 
Consider the following gate sequence,
\begin{equation}
    \mc C \coleq \mc U_q\mc G_q\mc U_{q-1}\mc G_{q-1}\cdots\mc U_1\mc G_1
\end{equation}
One can see that $\mc C(P_1) = P_1$. Now we design the following experiments parameterized by a positive integer $m$,
\begin{itemize}
    \item Initial state: $\rho_0 = (I+P_1)/2^n$,
    \item POVM measurement: $E_{\pm 1} = (I\pm P_1)/{2}$,
    \item Circuit: $\mc C^m = \left(\mc U_q\mc G_q\mc U_{q-1}\mc G_{q-1}\cdots\mc U_1\mc G_1\right)^m$.
\end{itemize}
Consider the outcome distribution generated by running these experiments within a noise model $\mc N$.
\begin{equation}
    \begin{aligned}
      p^{(m)}_{\pm 1}(\mc N) &= \Tr\left( \widetilde{E}_{\pm 1} \widetilde{\mc C}^m  (\widetilde{\rho}_0)\right) \\
      &= \Tr \left( \frac{I\pm P_1}{2} \cdot\left( \mc E^M \circ \left(\mc U_q\widetilde{\mc G}_q
      \cdots\mc U_1\widetilde{\mc G}_1\right)^m\circ\mc E^S \right)
      \left(\frac{I+P_1}{2^n}\right) \right)\\
      &= \Tr\left(\frac{I\pm P_1}{2} \cdot\frac{I+\lambda^M_{P_1}
      \left(\lambda^{\mc G_q}_{P_q}\cdots\lambda^{\mc G_2}_{P_2}\lambda^{\mc G_1}_{P_1}\right)^m
      \lambda^S_{P_1}P_1}{2^n} \right)\\
      &= \frac{1\pm\lambda^M_{P_1}
      \left(\lambda^{\mc G_q}_{P_q}\cdots\lambda^{\mc G_2}_{P_2}\lambda^{\mc G_1}_{P_1}\right)^m
      \lambda^S_{P_1}}{2}.
    \end{aligned}
\end{equation}
The expectation value is
\begin{equation}
    \mbb E^{(m)}(\mc N) = \lambda^M_{P_1}
      \left(\lambda^{\mc G_q}_{P_q}\cdots\lambda^{\mc G_2}_{P_2}\lambda^{\mc G_1}_{P_1}\right)^m
      \lambda^S_{P_1}.
\end{equation}
If we take the ratio of expectation values of two experiments with consecutive $m$, we obtain (recall that all these Pauli fidelities are strictly positive by Assumption 4)
\begin{equation}
    \mbb E^{m+1}(\mc N)/\mbb E^{m}(\mc N) = \lambda^{\mc G_q}_{P_q}\cdots\lambda^{\mc G_2}_{P_2}\lambda^{\mc G_1}_{P_1}.
\end{equation}
This implies that if two noise models have different values for the product of Pauli fidelities $\lambda^{\mc G_q}_{P_q}\cdots\lambda^{\mc G_2}_{P_2}\lambda^{\mc G_1}_{P_1}$, the above experiments would be able to distinguish between them. Therefore, $\lambda^{\mc G_q}_{P_q}\cdots\lambda^{\mc G_2}_{P_2}\lambda^{\mc G_1}_{P_1}$ is a learnable function. 
By taking the logarithm of this expression, we see that $f(\bm l)\coleq\sum_{k=1}^q l_{P_q}^{\mc G_q}$ is a learnable linear function of the logarithmic Pauli fidelities.
Notice that $f(\bm l)$ exactly corresponds to the circuit of $z$ according to the natural isometry between $F$ and $C_1(G)$. This tells us that every circuit in $G$ indeed corresponds to a learnable linear function. 
Combining with the fact that the circuits in $G$ span the cycle space $Z(G)$, and the fact that learnable functions are closed under linear combination (Lemma~\ref{le:learnable_is_space}), we conclude that $Z(G)\subseteq F_L$.

\medskip

$F_L\subseteq Z(G)$: 
For this part, we just need to show that $F_L$ is orthogonal to the cut space $U(G)$, which is the orthogonal complement of the cycle space $Z(G)$. To show this, we will construct a gauge transformation for each element of $U(G)$. The definition of learnability then requires a learnable linear function to be orthogonal to all gauge transformations, thus orthogonal to the entire cut space.

Consider a cut $V = V_1 \cup V_2$ (such that there is at least one edge between $V_1$ and $V_2$). We define the \emph{gauge transform map} $\mc M$ as the following Pauli diagonal map,
\begin{equation}
    \mc M(P) \coleq \left\{
        \begin{aligned}
          \eta P, \quad& \text{if}~\pt(P)\in V_1,\\
          P, \quad& \text{if}~\pt(P)\in V_2,
        \end{aligned}
    \right.\quad\forall P\in{\sf P}^n,
\end{equation}
for a positive parameter $\eta\ne 1$. The gauge transformation induced by $\mc M$ is defined in the same way as Eq.~\eqref{eq:gauge_trans}. 
We will show that there exists two noise models satisfying all the assumptions that are related by a gauge transformation (thus indistinguishable) but yields different values for the function corresponding to the cut $V_1\cup V_2$.

Starting with a noise model $\mc N = \{\mc E^S,\mc E^M, \Lambda\}$, we first calculate the gauge transformed noise model $\mc N'$.
The SPAM noise channels are transformed as
\begin{equation}\label{eq:spam_update_2}
    \mc E^{S'} = \mc M\mc E^S,\quad \mc E^{M'} = \mc E^M\mc M^{-1},
\end{equation}
which are still Pauli diagonal maps. 
Using exactly the same argument as in the proof of Theorem~\ref{th:nogo}, by choosing $\eta$ to be sufficiently close to $1$, these transformed maps are guaranteed to be CPTP and satisfy Assumption~4.

Secondly, the single-qubit unitaries are transformed as $\mc U' = \mc M\mc U\mc M^{-1}$. Calculate the following inner product for any $P,Q\in{\sf P}^n$,
\begin{equation}\label{eq:gauge_trans_2}
\begin{aligned}
    \Tr(P\cdot\mc U'(Q))&=\Tr(\mc M^\dagger(P)\cdot \mc U(\mc M^{-1}(Q)) )\\
    &= \eta^{\bm 1_{V_1}[\pt(P)]}(\eta^{-1})^{\bm 1_{V_1}[\pt(Q)]}\Tr(P\cdot\mc U(Q)).
\end{aligned}
\end{equation}
Here $\bm 1_{V_1}$ is the indicator function of $V_1$. 
We see that $\Tr(P\cdot\mc U'(Q))=\Tr(P\cdot\mc U(Q))$ if $\pt(P)=\pt(Q)$.
A crucial observation is that a product of single-qubit unitaries can never change the pattern of the input Pauli. More precisely, $\mc U(Q)$ is a linear combination of Pauli operators with the same pattern as $Q$. Therefore, if $\pt(P)\ne\pt(Q)$, we would have $\Tr(P\cdot\mc U'(Q))=\Tr(P\cdot\mc U(Q))=0$.
Combining the two cases, we conclude $\mc U'=\mc U$, \textit{i.e.}, the single-qubit unitaries are still noiseless in $\mc N'$.

Finally, the noisy Clifford gates are transformed as
\begin{equation}
    \begin{aligned}
        \widetilde{\mc G}'&= \mc M{\mc G}\Lambda_{\mc G}\mc M^{-1}\\
        &= \mc G\mc G^{-1}\mc M{\mc G}\Lambda_{\mc G}\mc M^{-1}\\
        &\eqcol \mc G\Lambda_{\mc G}'
    \end{aligned}
\end{equation}
where the transformed noise channel $\Lambda_{\mc G}'\coleq \mc G^{-1}\mc M{\mc G}\Lambda_{\mc G}\mc M^{-1}$ is a Pauli diagonal map. We now calculate its Pauli eigenvalues. For $P\in{\sf P}^n$,
\begin{equation}\label{eq:f_update_2}
    \begin{aligned}
        \Lambda_{\mc G}'(P) &=  \mc G^{-1}\mc M{\mc G}\Lambda_{\mc G}\mc M^{-1}(P)\\
        &=\eta^{\bm 1_{V_1}[\pt(\mc G(P))]}(\eta^{-1})^{\bm 1_{V_1}[\pt(P)]}\lambda^{\mc G}_P~P\\
        &=\left\{
        \begin{aligned}
            \eta \lambda_P^{\mc G},\quad& \pt(P)\in V_1,~\pt(\mc G(P))\in V_2.\\
            \eta^{-1} \lambda_P^{\mc G},\quad& \pt(P)\in V_2,~\pt(\mc G(P))\in V_1.\,\\
             \lambda_P^{\mc G},\quad& \text{otherwise}.\\ 
        \end{aligned}
        \right.
    \end{aligned}
\end{equation}
Again, Assumption 4 guarantees that $\Lambda_{\mc G}'$ is a CPTP map satisfying all of our noise assumptions as long as $\eta$ is sufficiently close to $1$. We omit the argument here as it is the same as in the previous proof.
Define $t_p \coleq \log \eta$ where $p$ denotes the cut $V_1\cup V_2$. The above gauge transformation of the log Pauli fidelity can be written as
\begin{equation}
    \bm l' = \bm l + t_p \bm v_p
\end{equation}
where $\bm v_p$ is the cut vector of $V = V_1\cup V_2$ as defined in Eq.~\eqref{eq:cut}.

\medskip

We have just defined a gauge transformation $\mc M_p$ for an arbitrary cut $p$.
Fix a basis of the cut space $B$ (where vectors in $B$ has the form in Eq.~\eqref{eq:cut}). For a generic element of the cut space $\bm v\in U(G)$, we can decompose it as $\bm v = \sum_{p\in B} t_p \bm v_p$ ($t_p\in\mathbb{R}$). We define the gauge transformation $\mc M_{\bm v}$ associated with $\bm v$ as a consecutive application of the gauge transformations $\{\mc M_p\}$ for each $p\in B$, each with parameter $t_p$. Here we assume that each $|t_p|$ is sufficiently small, as otherwise we can rescale the vector. This implies that $\mc M_{\bm v}$ is a valid gauge transformation.
The effect of such a transformation is
\begin{equation}
    \bm l' = \bm l + \bm v.
\end{equation}

Now, Definition~\ref{de:learnability} implies that a learnable function $\bm f$ must remain unchanged under gauge transformations (as they result in indistinguishable noise models), which means that $\bm f\cdot \bm l' = \bm f\cdot \bm l$. Thus, for all $\bm f\in F_L$, and all $\bm v \in U(G)$, we must have
\begin{equation}
    \bm f\cdot \bm v = \bm f\cdot \bm l' - \bm f\cdot \bm l = 0.
\end{equation}
That is, $F_L$ must be orthogonal to the cut space $U(G)$. According to Lemma~\ref{le:complement}, $Z(G)$ is the orthogonal complement of $U(G)$, so we conclude that $F_L\subseteq Z(G)$. This completes the second part of our proof. 

\end{proof}

\subsection{Learnability under no-crosstalk assumption}\label{sec:no_crosstalk}
    As we commented before, the way we define the gate noise captures a general form of crosstalk~\cite{sarovar2020detecting}. One may ask, if we further make a favorable assumption that gate noise has no crosstalk, would this make the learning of noise much easier.
    To consider this rigorously, we introduce the following optional assumption. See Fig.~\ref{fig:crosstalk} for an illustration.
    \begin{itemize}
        \item \textbf{Assumption 5} (No crosstalk.) For any $\mc G\in\mf G$ that acts non-trivially only on a $k$-qubit subspace, the associated Pauli noise channel also acts non-trivially only on that subspace. In other words, if $\mc G = \mc G'\otimes \mc I$, we have $\widetilde{\mc G} = \left( \mc G'\circ\Lambda_{\mc G} \right)\otimes \mc I$ where $\Lambda_{\mc G}$ is an $k$-qubit Pauli channel
        depending only on $\mc G$ and the (ordered) subset of qubits on which $\mc G$ acts.
    \end{itemize}
    
    \begin{figure}[!htp]
        \centering
        \includegraphics[width=\columnwidth]{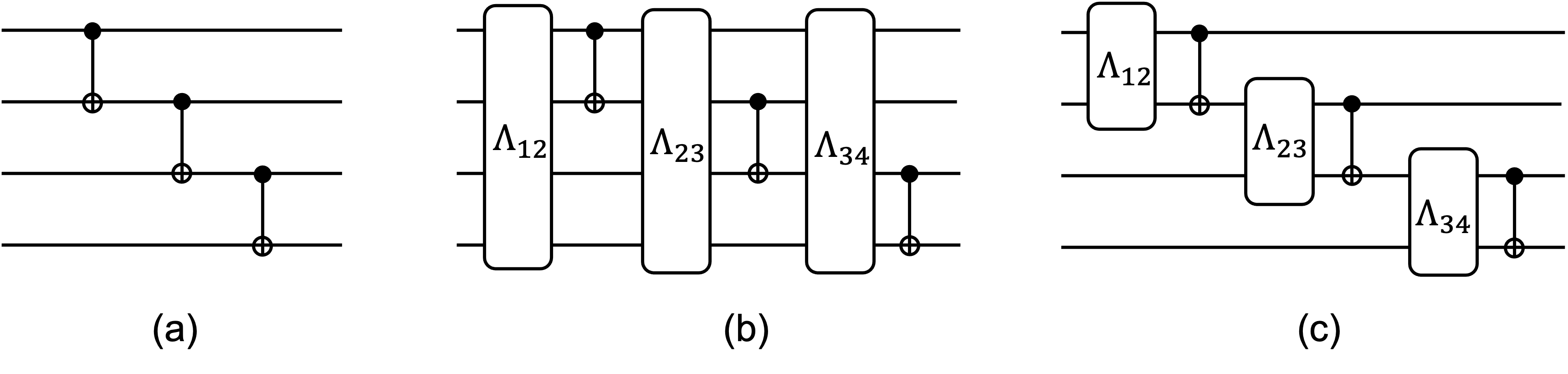}
        \caption{Illustration of the crosstalk model. (a) A $4$-qubit circuit consists of three ideal CNOT gates. (b) Full crosstalk. The noise channels are $4$-qubit and depends on the qubits the CNOT acts on. (c) No crosstalk. The noise channel only acts on a $2$-qubit subspace. It can still depend on the qubits the CNOT acts on.}
        \label{fig:crosstalk}
    \end{figure}

    Assumption 5 reduces the number of independent parameters of a noise model. 
    One might expect certain unlearnable functions to become learnable with this assumption.
    Here, we show that the simple criteria of learnablity given in Theorem~\ref{th:nogo} still hold even in this case, as stated in the following proposition.
    
    \begin{proposition}\label{prop:nogo_nocross}
        With Assumption 1-5, for any $k$-qubit Clifford gate $\mc G$ and Pauli operator $P_a$, the Pauli fidelity $\lambda_a^{\mc G}$ is unlearnable if and only if $\mc G$ changes the pattern of $P_a$, \textit{i.e.}, $\pt(\mc G(P_a))\ne \pt(P_a)$. 
    \end{proposition}
    \begin{proof}
        We just need to modify the proof of Theorem~\ref{th:nogo}.
        For the ``only if'' part, restrict our attention to the $k$-qubit subsystem that $\mc G$ acts on, and do a cycle benchmarking protocol as in the original proof. We can easily conclude that $\lambda_a^{\mc G}$ is learnable if $\pt(P_a) = \pt(\mc G(P_a))$.
        
        For the ``if'' part, construct the same gauge transformation map as in the original proof.
        That is, for an index $i\in[n]$ such that $\pt(P_a)_i\ne \pt(\mc G(P_a))_i$, let $\mc M = \mc D_i\otimes\mc I_{[n]\backslash i}$ where $D_i$ is the single-qubit deplorizing channel on the $i$th qubit with some parameter $\eta$. 
        With the no-crosstalk assumption, a generic $k$-qubit noisy Clifford gate $\widetilde{\mc T}$ transforms as
        \begin{equation}
            \widetilde{\mc T}\otimes\mc I \mapsto \mc M\circ (\widetilde{\mc T}\otimes\mc I)\circ \mc M^{-1}.
        \end{equation}        
        If $\mc T$ does not act on the $i$th qubit, $\mc M$ commutes with $\widetilde{\mc T}$ and the noisy Clifford gate remains unchanged. 
        If $\mc T$ acts non-trivially on the $i$th qubit,
        \begin{equation}
            \widetilde{\mc T}\otimes\mc I \mapsto  (\mc D_i\circ\widetilde{\mc T}\circ\mc D_i^{-1})\otimes\mc I.
        \end{equation}
        This means the transformed noise channel acts non-trivially only on the $k$-qubit subsystem that $\mc G$ acts on, thus satisfies the no-crosstalk assumption.
        The Pauli fidelities of the noise channel will be updated as Eq.~\eqref{eq:paulifidelityupdate}.
        Following the same argument of the original proof, we conclude that $\lambda_{a}^{\mc G}$ is unlearnable if $\pt(P_a) \neq \pt(\mc G(P_a))$.
    \end{proof}

    It is also possible to generalize the graph theoretical characterization in Theorem~\ref{th:space} to the no-crosstalk case.
    One challenge in this case is that, different edges in the pattern transfer graph no longer stand for independent variables.
    For example, consider a $3$-qubit system and a CNOT on the first two qubits.
    Since $\mr{CNOT}(XI) = XX$, we would have the following two edges in the pattern transfer graph 
    $$e_{XII,\mr{CNOT}\otimes\mc I} = (100,110),\quad e_{XIX,\mr{CNOT}\otimes\mc I} = (101,111).$$
    However, with the no-crosstalk assumption, we have 
    \begin{equation}
        \lambda_{XII}^{\mr{CNOT}\otimes \mc I} = 
        \lambda_{XIX}^{\mr{CNOT}\otimes \mc I} =
        \lambda_{XI}^{\mr{CNOT}},
    \end{equation}
    which means the above two edges represent the same Pauli fidelity.
    As a result, a gauge transformation (as defined in the proof of Theorem~\ref{th:space}) that changes $\lambda_{XII}$ and $\lambda_{XIX}$ differently is no longer a valid transformation.
    In other word, a cut represents a valid gauge transformation only if it cuts through all the edges for the same Pauli fidelity simultaneously.
    This could decrease the number of unlearnable degrees of freedom.
    We leave the precise characterization of the learnable space with no-crosstalk assumptions as an open question.
    It is also interesting to study the learnability under other practical assumptions about the Pauli noise model, such as the sparse Pauli-Lindbladian model~\cite{berg2022probabilistic} and the Markovian graph model~\cite{flammia2020efficient,harper2020efficient}.

    \subsection{Learnability of Pauli error rates}
    We have been focusing on the learnability of Pauli fidelities $\bm\lambda$. One may ask similar questions about Pauli error rates $\bm p$. 
    It turns out that, at least in the weak-noise regime (\textit{i.e.}, $\lambda_a$ close to $1$), the learnability of $\bm p$ is $\bm \lambda$ are highly related. To see this, note that
    \begin{equation}
        \begin{aligned}
            p_a &= \frac{1}{4^n}\sum_{b}(-1)^\expval{a,b}\lambda_b\\
            &\approx \frac{1}{4^n}\sum_{b}(-1)^\expval{a,b}(\log\lambda_b + 1)\\
            &=\frac{1}{4^n}\sum_{b}(-1)^\expval{a,b} l_b + \delta_{a,\bm 0},
        \end{aligned}
    \end{equation}
    which means that $p_a$ is approximately a linear function of the logarithmic Pauli fidelity $\bm l$. 
    Therefore, one can in principle use Theorem~\ref{th:space} to completely decide the learnability of any Pauli error rates (with weak-noise approximation).
    Furthermore, since the Walsh-Hadamard transformation is invertible, different $p_a$ corresponds to linearly-independent function of $\bm l$. 
    This means that the number of linearly independent equations we can obtain about the Pauli error rates is the same as the learnable degrees of freedom of the Pauli fidelities.
    In Table~\ref{tab:CNOT_full}, we list a basis for all the learnable Pauli fidelities/Pauli error rates. One can see that there is an exact correspondence between these two. We leave a fully general argument for future study.
    
    \begin{table}[!htp]
        \centering
        \begin{tabular}{|c|c|}
            \hline
            Learnable log Pauli fidelities  & $l_{II},l_{ZI},l_{IX},l_{ZX},
            l_{XZ},l_{YY},l_{XY},l_{YZ},$ 
            \\&$l_{IZ}+l_{ZZ},l_{IY}+l_{ZY},l_{IZ}+l_{ZY},l_{XI}+l_{XX},l_{YI}+l_{YX},l_{XI}+l_{YX}$ \\
            \hline
            Learnable Pauli error rates  & $p_{II},p_{ZI},p_{IX},p_{ZX},
            p_{XZ},p_{YY},p_{XY},p_{YZ},$ 
            \\ (approximately)&$p_{IZ}+p_{ZZ},p_{IY}+p_{ZY},p_{IZ}+p_{ZY},p_{XI}+p_{XX},p_{YI}+p_{YX},p_{XI}+p_{YX}$ \\
            \hline
        \end{tabular}
        \caption{A complete basis for the learnable linear functions of log Pauli fidelities and Pauli error rates (the latter is approximate) for a single CNOT gate.}
        \label{tab:CNOT_full}
    \end{table}

\section{Additional details about the numerical simulations}\label{sec:numerics}

In this section, we provide more details about the numerical simulations mentioned in the main text. The simulation is conducted using \texttt{qiskit}~\cite{Qiskit}, an open-source Python package for quantum computing. We simulate a two-qubit system where single-qubit Clifford gates are noiseless, and CNOT is subject to amplitude damping channels on both qubits. Note that amplitude damping is not Pauli noise, but we apply randomized compiling and will only estimate its Pauli diagonal part. We also note that, \texttt{qiskit} adds the noise channel \emph{after} gate by default, but our theory assume the noise to be \emph{before} gate. These two models can be easily converted between each other via
\begin{equation}
    \mc G\circ\Lambda_{\mc G} = (\mc G\circ\Lambda_{\mc G}\circ\mc G^{\dagger})\circ\mc G = \Lambda_{\mc G}'\circ\mc G.
\end{equation}
If $\mc G$ is Clifford, $\Lambda_{\mc G}$ is a Pauli channel if and only if $\Lambda_{\mc G}'$ is a Pauli channel. In the following, we will be consistent with our theory and assume the noise to be before gate.
Besides, we let the measurement to have $0.3\%$ bit-flip rate on each qubit and the state-preparation to be noiseless.

Fig.~\ref{fig:main_sim_cbraw} shows the estimates collected using standard CB and interleaved CB (circuits shown in Fig.~1 of main text). Compared to the true values, we see that both simulations yields accurate predictions of the learnable Pauli fidelities.

\begin{figure}[!htp]
    \centering
    \includegraphics[width=\linewidth]{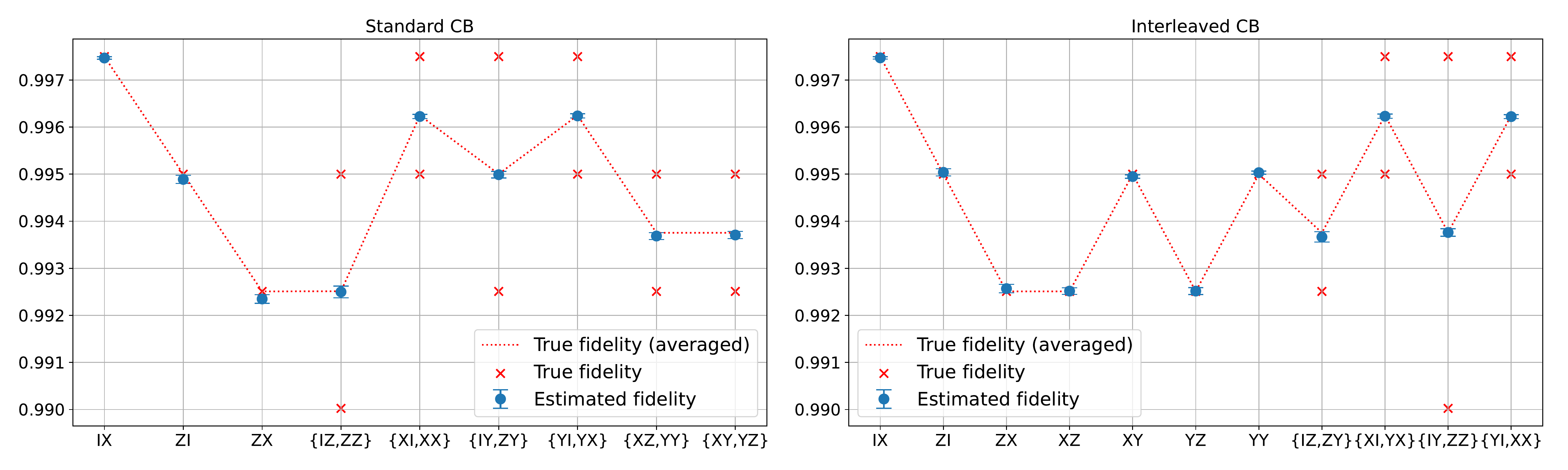}
    \caption{Numerical estimates of Pauli fidelities of a CNOT gate via standard CB (left) and CB with interleaved gates (right), using circuits shown in Fig.~1 of main text. Each Pauli fidelity is fitted using seven different circuit depths $L=[2,2^2,...,2^7]$. For each depth $C=30$ random circuits and $200$ shots of measurements are used.
    The red cross shows the true fidelities and the red dash line shows the average of true fidelities within any two-Pauli group.
    }
    \label{fig:main_sim_cbraw}
\end{figure}

Fig.~\ref{fig:main_sim_cbfeasible} (a) calculates the physically feasible region according to the estimates in terms of $\{\lambda_{XX},\lambda_{ZZ}\}$, using approaches discussed in the main text.
Due to the special structure of the twirled amplitude damping noise (no $Z$-error), the feasible region for $\lambda_{XX}$ is extremely narrow. To eliminate the effect of statistical error, we allow a smoothing parameter $\varepsilon$ in calculating the physical region, making the constraints to be $p_a\ge-\varepsilon$. Here $\varepsilon$ is chosen to be the largest standard deviation in estimating the learnable Pauli fidelities. In Fig.~\ref{fig:main_sim_cbfeasible} (b)(c) we see that the true fidelity indeed falls into the physical region and is actually close to the lower-left corner of the physical region. 

\begin{figure}[!htp]
    \centering
    \subfloat[feasible region]{
    \centering
    \includegraphics[width=0.4\linewidth]{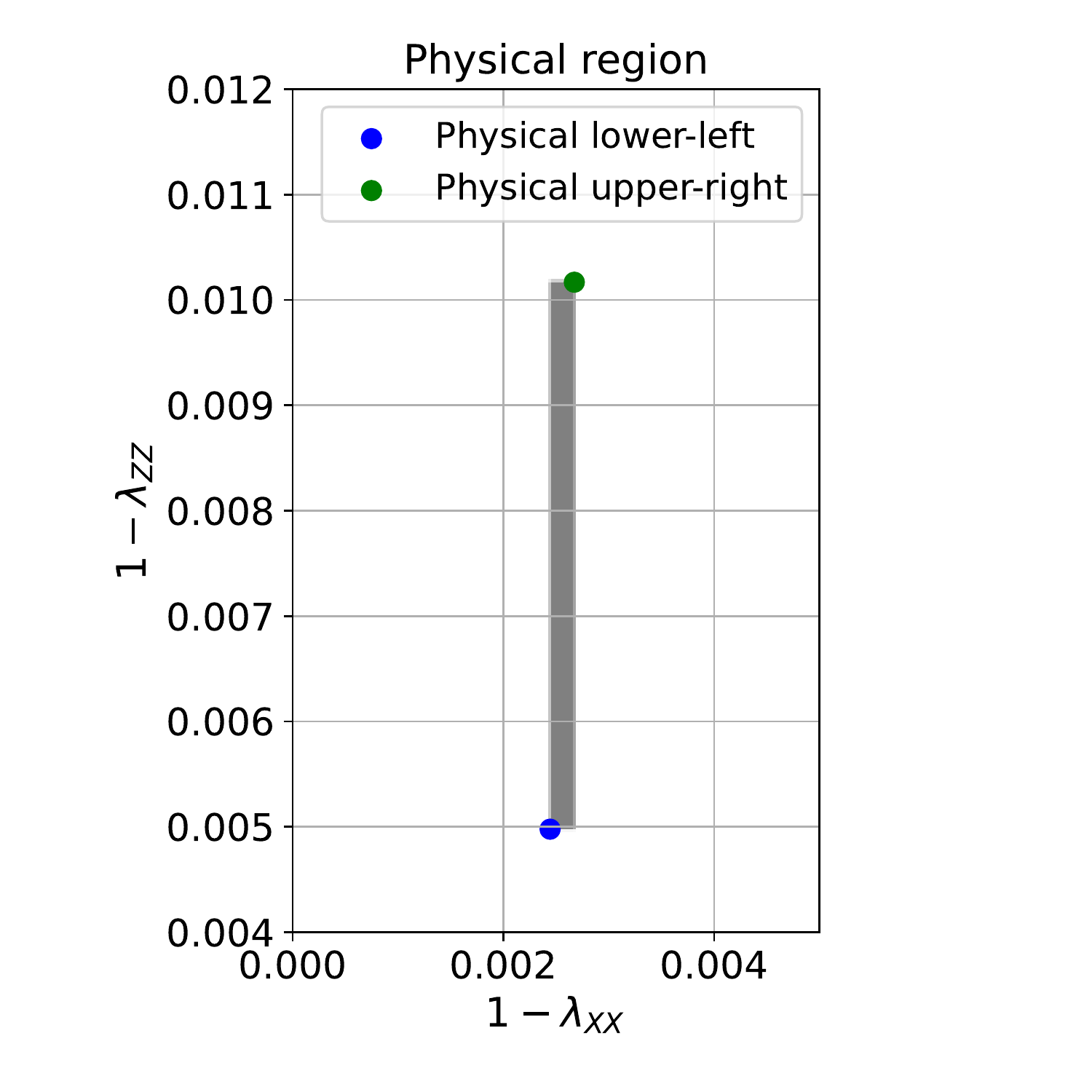}
    }\\
    \subfloat[Pauli fidelities]{
    \centering
    \includegraphics[width=0.5\linewidth]{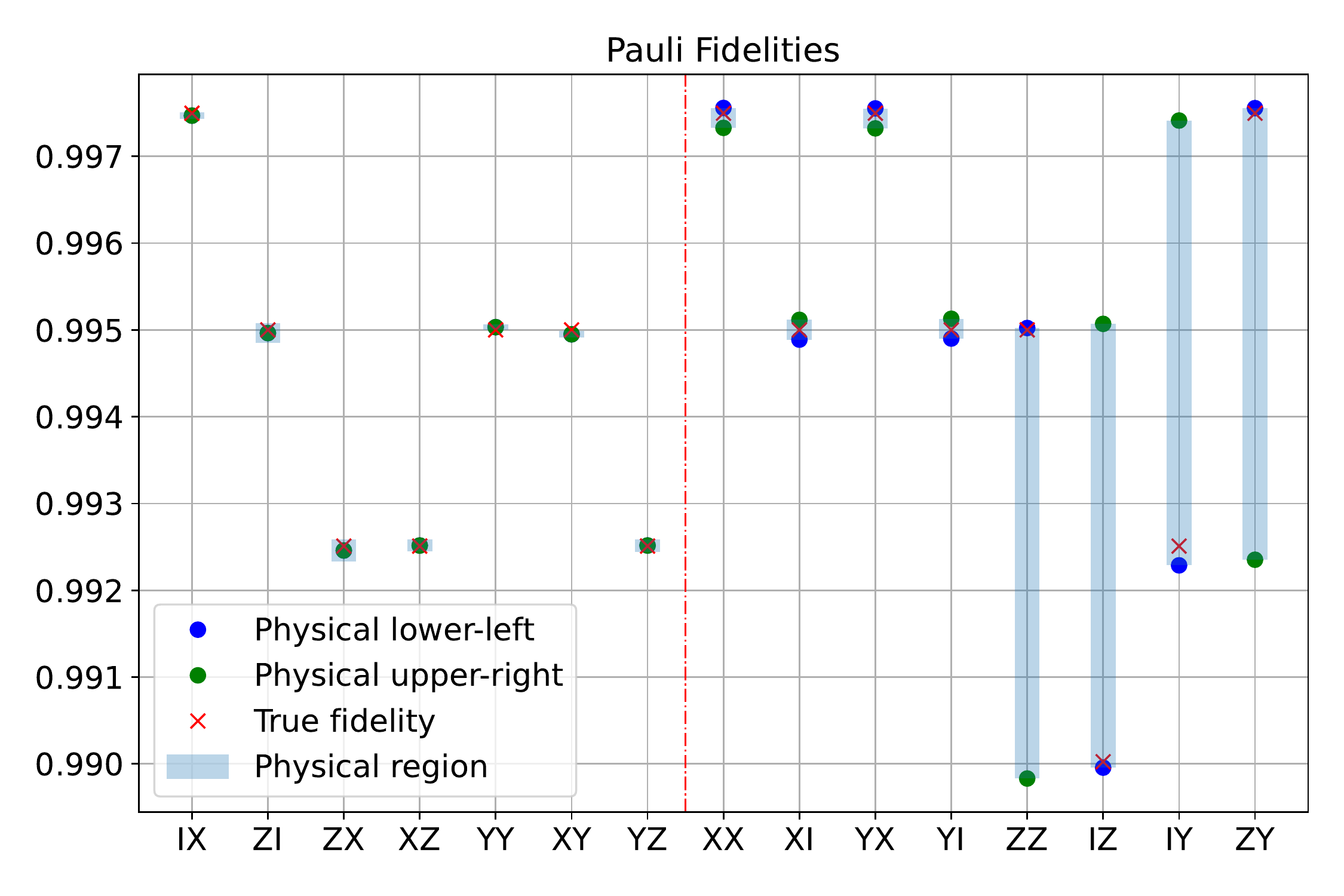}
    }
    \subfloat[Pauli errors]{
    \centering
    \includegraphics[width=0.5\linewidth]{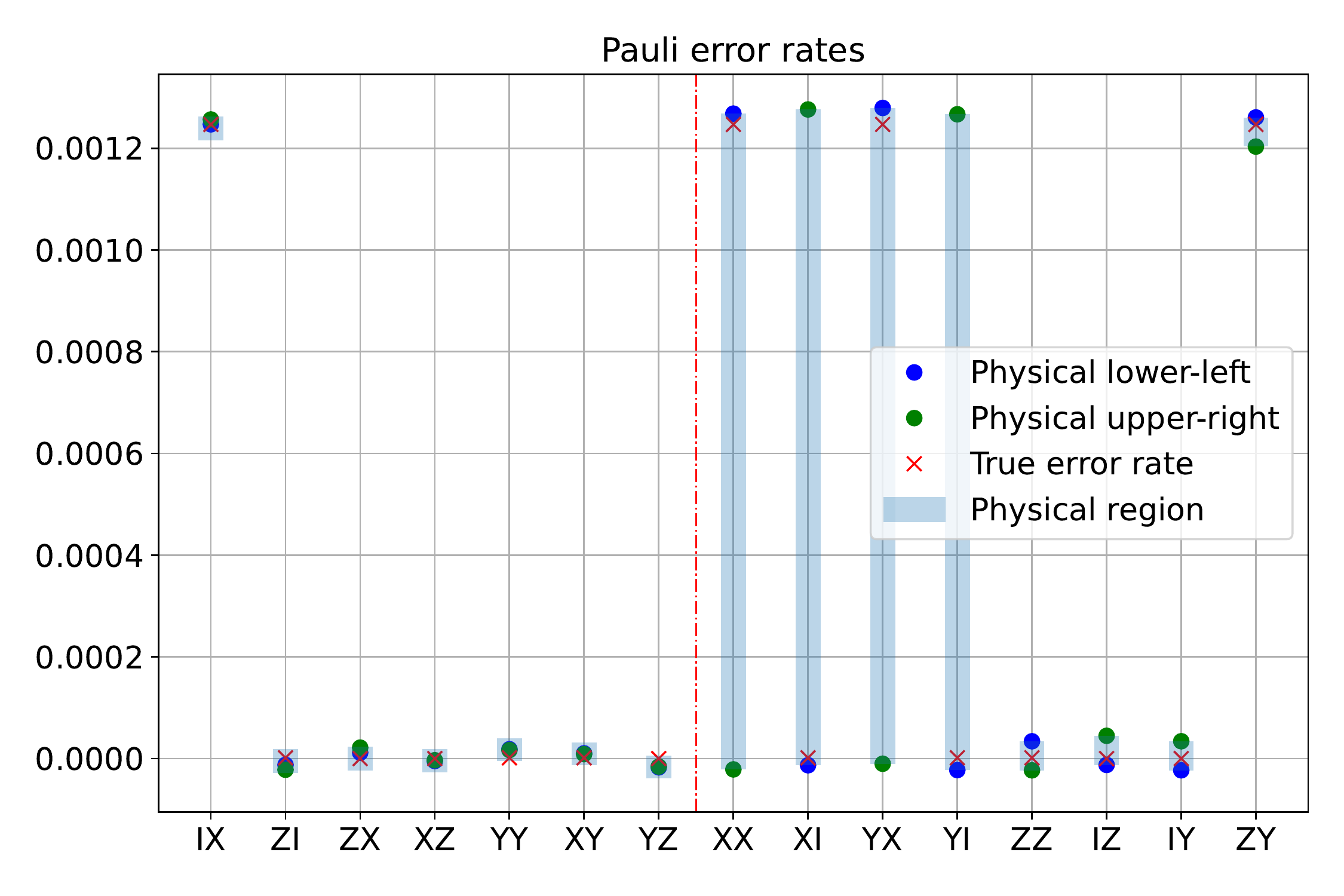}
    }
    \caption{Feasible region of the learned Pauli noise model, using data from Fig.~\ref{fig:main_sim_cbraw}. (a) Feasible region of the unlearnable degrees of freedom in terms of $\lambda_{XX}$ and $\lambda_{ZZ}$. (b) Feasible region of individual Pauli fidelities. (c) Feasible region of individual Pauli errors.}
    \label{fig:main_sim_cbfeasible}
\end{figure}

Fig.~\ref{fig:app_sim_intercept} shows the simulation results of intercept CB. We see that, we obtain an accurate estimate even for the unlearnable Pauli fidelities. Besides, the estimate lies inside the physically feasible region up to a standard deviation.
This shows that intercept CB should work well in resolving the unlearnability if we do have access to noiseless state-preparation (and the method is robust against measurement noise). Therefore, failure of this method in experiment implies a non-negligible state-preparation error, as discussed in the main text.

\begin{figure}[!htp]
    \centering
    \subfloat[]{
    \centering
    \includegraphics[width=0.5\linewidth]{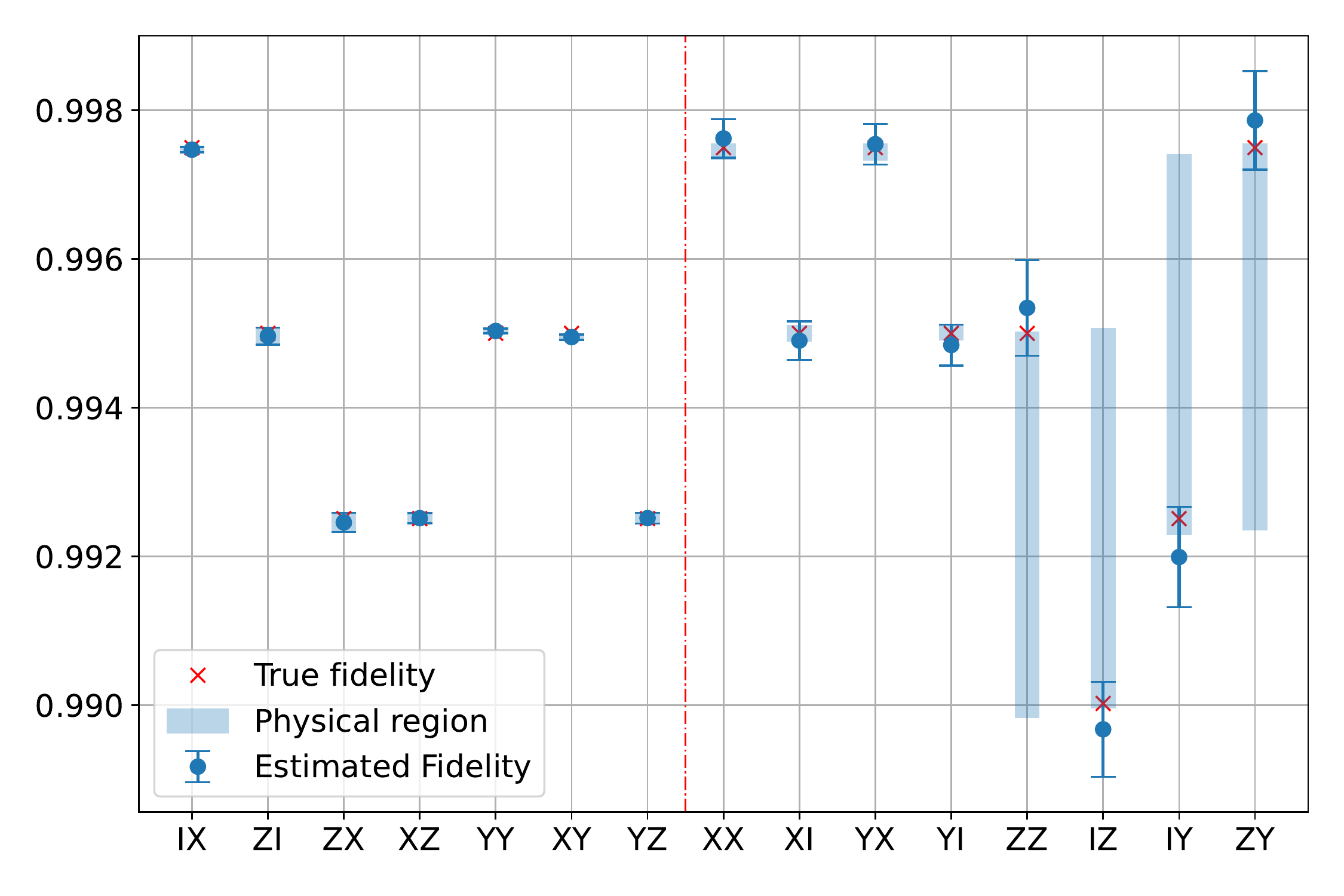}
    }
    \subfloat[]{
    \centering
    \includegraphics[width=0.5\linewidth]{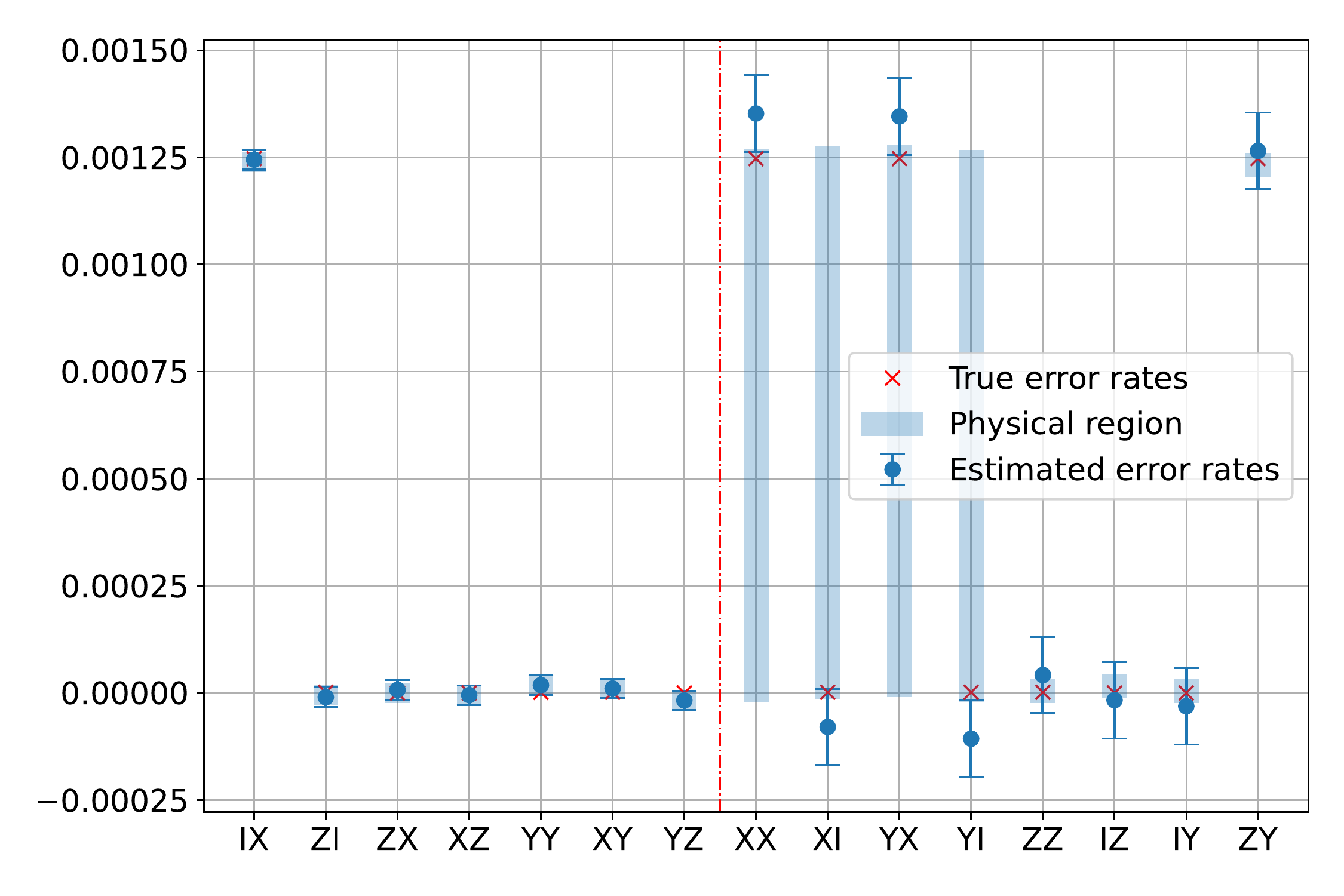}
    }
    \caption{The learned Pauli noise model using intercept CB. The feasible region (blue bars) are taken from Fig.~\ref{fig:main_sim_cbfeasible}. Estimates of Pauli fidelities (a) and Pauli error rates (b). Each data point is fitted using seven different circuit depths $L=[2,2^2,...,2^7]$. For each depth $C=300$ random circuits and $2000$ shots of measurements are used.
    }
    \label{fig:app_sim_intercept}
\end{figure}

\section{Justification for the claim in Sec.~\ref{sec:space}}\label{sec:justification}

We claim in Sec.~\ref{sec:space} that any measurement probability generated in experiment can be expressed as a polynomial of Pauli fidelities, and that each term in the polynomial can be learned in a CB experiment. This is the motivation why we only care for a single monomial of Pauli fidelities. Here we justify this claim.

Consider the most general experimental design: prepare some initial state $\rho_0$, apply some quantum circuit $\mc C$, and conduct a POVM measurement $\{E_j\}_j$. Denote the noisy realization of these objects with a tilde. Because of noise, the probability of obtaining a certain measurement outcome $j$ is
\begin{equation}
    \mr{Pr}(j) = \Tr\left( \widetilde{E}_j\widetilde{\mc C}(\widetilde{\rho}_0) \right) = \Tr\left( E_j\left(\Lambda^M\circ\widetilde{\mc C}\circ\Lambda^S\right)(\rho_0) \right) \equiv \Tr\left( E_j\rho' \right).
\end{equation}
Here $\Lambda^S,\Lambda^M$ are the noise channels for state preparation and measurement, respectively. The Pauli fidelity of them are denoted by $\lambda_a^S,\lambda_a^M$ for Pauli operator $a$, respectively. We define $\rho'\coleq (\Lambda^M\circ\widetilde{\mc C}\circ\Lambda^S)(\rho_0)$ which encodes all the information that can be extracted from a quantum measurements. We will obtain a general formula for $\rho'$.

First note that a general noisy quantum circuit $\widetilde{\mc C}$ satisfying our assumptions can be expressed as
\begin{equation}
    \widetilde{\mc C} = C\od m \circ \widetilde{\mc G}_{{m}} \circ \cdots \circ C\od 1 \circ \widetilde{\mc G}_{{1}} \circ C\od 0,
\end{equation}
where ${\mc G}_{{j}}\in{\mf G}$ is an $n$-qubit Clifford gate and $C\od j$ is the tensor product of single-qubit gates. A crucial property for single-qubit gates is that they never change the Pauli pattern. More rigorously, one have that
\begin{equation}
    C\od{j}(P_a) = \sum_{b\sim \pt(a)}c_{b,a}\od{j}P_b,\quad\forall P_a\in{\sf P}^n,
\end{equation}
where $c_{b,a}\od{j}\in\mbb R$, and the summation is over all $P_b$ that have the same Pauli pattern as $P_a$.

\medskip

\noindent Now consider the action of $\widetilde{\mc C}$ on an arbitrary Pauli operator $P_a$. 
\begin{equation}
    \begin{aligned}
        \widetilde{\mc C}(P_a) &= (C\od m \circ \widetilde{\mc G}_{{m}} \circ \cdots \circ C\od 1 \circ \widetilde{\mc G}_{{1}} \circ C\od 0)(P_a)\\
        &= (C\od m \circ \widetilde{\mc G}_{{m}} \circ \cdots \circ C\od 1 \circ \widetilde{\mc G}_{{1}}) \left(\sum_{b_0\sim \pt(a)} c_{b_0,a}\od{0} P_{b_0} \right)\\
        &= (C\od m \circ \widetilde{\mc G}_{{m}} \circ \cdots \circ C\od 1) \left(\sum_{b_0\sim \pt(a)} c_{b_0,a}\od 0\lambda_{b_0}^{\mc G_{1}} P_{\mc G_{1}(b_0)} \right)\\
        &= (C\od m \circ \widetilde{\mc G}_{{m}} \circ \cdots \circ C\od 2) \left(\sum_{\substack{
        b_0\sim \pt(a),\\
        b_1\sim \pt(\mc G_{1}(b_0))
        }} c_{b_1,\mc G_{1}(b_0)}\od{1}c_{b_0,a}\od{0} \lambda_{b_1}^{\mc G_{2}}\lambda_{b_0}^{\mc G_{1}} P_{\mc G_{2}(b_1)} \right)\\
        &= \cdots\\
        &= \sum_{\substack{
        b_0\sim \pt(a),\\
        b_1\sim \pt(\mc G_{1}(b_0)),\\
        \dots\\
        b_m\sim \pt(\mc G_{m}(b_{m-1}))
        }}c_{b_m,\mc G_{m}(b_{m-1})}\od{m}\cdots c_{b_1,\mc G_{1}(b_0)}\od{1} c_{b_0,a}\od{0} \lambda_{b_{m-1}}^{\mc G_{m}}\cdots\lambda_{b_1}^{\mc G_{2}}\lambda_{b_0}^{\mc G_{1}} P_{b_m}.
    \end{aligned}
\end{equation}
For any initial state $\rho_0$, we can decompose it via Pauli operators as 
\begin{equation}
    \rho_0 = \frac{1}{2^n} I + \sum_{a\ne\bm 0}\alpha_aP_a.
\end{equation}
Going through the state preparation noise, the quantum circuit, and the measurement noise, the state evolves to
\begin{equation}
    \begin{aligned}\label{eq:rho'2new}
        \rho' &= (\Lambda^{M}\circ\widetilde{\mc C}\circ \Lambda^{S})(\frac{1}{2^n}I + \sum_{a\ne\bm 0}\alpha_aP_a)\\
        &= \frac{1}{2^n}I +\sum_{a\ne\bm 0}\alpha_a\sum_{\substack{
        b_0\sim \pt(a),\\
        b_1\sim \pt(\mc G_{1}(b_0)),\\
        \dots\\
        b_m\sim \pt(\mc G_{m}(b_{m-1}))
        }}c_{b_m,\mc G_{m}(b_{m-1})}\od{m}\cdots c_{b_1,\mc G_{1}(b_0)}\od{1} c_{b_0,a}\od{0} ~\lambda_{\pt(b_m)}^M\lambda_{b_{m-1}}^{\mc G_{m}}\cdots\lambda_{b_1}^{\mc G_{2}}\lambda_{b_0}^{\mc G_{1}}\lambda_{\pt(a)}^S P_{b_m}\\
        &\equiv\frac{1}{2^n}I +\sum_{a\ne\bm 0}\alpha_a\sum_{\substack{
        b_0\sim \pt(a),\\
        b_1\sim \pt(\mc G_{1}(b_0)),\\
        \dots\\
        b_m\sim \pt(\mc G_{m}(b_{m-1}))
        }}c_{b_m,\mc G_{m}(b_{m-1})}\od{m}\cdots c_{b_1,\mc G_{1}(b_0)}\od{1} c_{b_0,a}\od{0} ~\Gamma_{\bm b,a} P_{b_m}.
    \end{aligned}
\end{equation}
Here we define $\Gamma_{\bm b,a}=\lambda_{\pt(b_m)}^M\lambda_{b_{m-1}}^{\mc G_{m}}\cdots\lambda_{b_1}^{\mc G_{2}}\lambda_{b_0}^{\mc G_{1}}\lambda_{\pt(a)}^S$, which is a monomial of Pauli fidelities.
The measurement outcome probability $\mr{Pr}(j)$ is a linear combination of such $\Gamma_{\bm b,a}$ plus some constant.
Moreover, each $\Gamma_{\bm b,a}$ of the above form can also be learned from a simple experiment, by choosing the initial state to be a $+1$ eigenstate of $P_a$, measurement operator to be $P_{b_m}$, and $C^{(j)}$ to be the product of single-qubit Clifford gates satisfying $C^{(j)}(\mc G_j({b_{j-1}})) = {b_{j}}$ (which is possible because $\pt(b_j)=\pt(\mc G_j(b_{j-1}))$).
Therefore, to completely characterize a noise model, we only need to extract the products of Pauli fidelities in the form of $\Gamma_{\bm b,a}$. This justifies our earlier claim.

\end{appendix}

\end{document}